\documentclass[a4paper]{article}
\usepackage[T1]{fontenc}
\usepackage[french,english]{babel}
\usepackage{multicol}
\usepackage{amssymb,amsmath}
\usepackage{float}
\usepackage{theorem}
\usepackage{functan}
\usepackage{graphicx,epsfig}

\newcommand{\argmin}{\mathop{\mathrm{argmin}}}
\newcommand{\argmax}{\mathop{\mathrm{argmax}}}

\makeatletter
\DeclareRobustCommand{\qed}{%
  \ifmmode 
  \else \leavevmode\unskip\penalty9999 \hbox{}\nobreak\hfill
  \fi
  \quad\hbox{\qedsymbol}}
\newcommand{\openbox}{\leavevmode
  \hbox to.77778em{%
  \hfil\vrule
  \vbox to.675em{\hrule width.6em\vfil\hrule}%
  \vrule\hfil}}
\newcommand{\qedsymbol}{\openbox}
\newenvironment{proof}[1][Proof]{\par
  \normalfont
  \topsep6\p@\@plus6\p@ \trivlist
  \item[\hskip\labelsep\bfseries
    #1\@addpunct{.}]\ignorespaces
}{%
  \qed\endtrivlist
}
\makeatother

\numberwithin{equation}{section}

\Macro{C}{\mathcal{C}}
\Macro{C0}{\mathcal{C}_0(G)}
\Macro{C1}{\mathcal{C}^1}
\Macro{C2}{\mathcal{C}^2}
\Macro{Ci}{\mathcal{C}^{\infty}}

\Macro{R}{\mathbb{R}}
\Macro{N}{\mathbb{N}}
\Macro{Z}{\mathbb{Z}}

\Macro{Pb}{\mathcal{P}}

\Macro{P}{\mathbb{P}}
\Macro{esp}{\mathbb{E}}
\Macro{var}{\mathbb{V}}
\Macro{nor}{\mathcal{N}}

\Macro{L}{\mathcal{L}}

\Macro{F}{\mathcal{F}}

\Macro{ind}{\mathbf{1}}

\Macro{I}{\mathcal{I}}

\newnorm{matrix}{\delcasesize{#1}%
{\left|\left|\left| #2\right|\right|\right|}%
{|||#2|||}%
{#1|#1|#1|#2 #1|#1|#1|}}

\newtheorem{thm}{Theorem}[section]

\newtheorem{prop}{Proposition}[section]
\newtheorem{lem}{Lemma}[section]

{\theorembodyfont{\rmfamily} }

\title{Adaptive optimal allocation in stratified sampling methods}
\author{ Pierre Etor\'e\thanks{CERMICS, Universit\'e Paris Est, 6-8 avenue Blaise Pascal, Cit\'e
Descartes, Champs-sur-Marne, 77455 Marne la Vall\'ee Cedex 2,
e-mail : etore@cermics.enpc.fr, supported by the ANR project ADAP'MC},   Benjamin Jourdain\thanks{project team Mathfi, CERMICS, Universit\'e Paris Est, 6-8 avenue Blaise Pascal, Cit\'e
Descartes, Champs-sur-Marne, 77455 Marne la Vall\'ee Cedex 2,
e-mail : jourdain@cermics.enpc.fr}}
\begin{document}
\maketitle
\begin{abstract}
In this paper, we propose a stratified sampling algorithm in which the random
drawings made in the strata to compute the expectation of interest are
also used to adaptively modify the proportion of further drawings in each
stratum. These proportions converge to the optimal allocation in terms
of variance reduction. And our stratified estimator is asymptotically
normal with asymptotic variance equal to the minimal one. Numerical
experiments confirm the efficiency of our algorithm.
  \end{abstract}

\section*{Introduction}
Let $X$ be a $\m{R}^d$-valued random variable and $f:\m{R}^d\rightarrow\m{R}$ a measurable function such that $\m{esp}(f^2(X))<~\infty$.
We are interested in the computation of $c=\m{esp}(f(X))$ using a
stratified sampling Monte-Carlo estimator.
We suppose that $(A_i)_{1\leq i\leq I}$ is a partition of $\m{R}^{d}$
into $I$ {\it strata} such that $p_i=\m{P}[X\in A_i]$ is known
explicitely for $i\in\{1,\hdots,I\}$. Up to removing some strata, we
assume from now on that $p_i$ is positive for all
$i\in\{1,\hdots,I\}$. The stratified Monte-Carlo estimator of $c$ (see
\cite{glass2} p.209-235 and the references therein for a presentation
more detailed than the current introduction) is
based on the equality $\m{esp}(f(X))=\sum_{i=1}^I
p_i\m{esp}(f(X_i))$ where $X_i$ denotes a random variable distributed
according to the conditional law of $X$ given $X\in A_i$. Indeed, when
the variables $X_i$ are simulable, it is possible
to estimate each expectation in the right-hand-side using $N_i$ i.i.d
drawings of $X_i$. Let $N=\sum_{i=1}^IN_i$ be the total number of
drawings (in all the strata) and $q_i=N_i/N$ denote the proportion of drawings made in stratum $i$.

Then $\widehat{c}$ is defined by
$$\widehat{c}=\sum_{i=1}^I\frac{p_i}{N_i}\sum_{j=1}^{N_i}f(X_i^j)=
\frac{1}{N}\sum_{i=1}^I\frac{p_i}{q_i}\sum_{j=1}^{q_iN}f(X_i^j),$$
where for each $i$ the $X_i^j$'s, $1\leq j\leq N_i$, are distributed as~$X_i$,
and all the $X_i^j$'s, for  $1\leq i\leq I$, $1\leq j\leq N_i$ are drawn independently.
This stratified sampling estimator can be implemented for instance when
$X$ is distributed according to the Normal law on $\m{R}^d$,
$A_i=\{x\in\m{R}^d:y_{i-1}<u'x\leq y_i\}$ where
$-\infty=y_0<y_1<\hdots<y_{I-1}<y_I=+\infty$ and $u\in\m{R}^d$ is such
that $|u|=1$. Indeed, then one has $p_i=N(y_i)-N(y_{i-1})$ with $N(.)$
denoting the cumulative distribution function of the one dimensional
normal law and it is easy to simulate according to the conditional law
of $X$ given $y_{i-1}<u'X\leq y_i$ (see section \ref{ssopt} for a
numerical example in the context of options pricing).
\vspace{0.2cm}
We have $\m{esp}(\widehat{c})=c$ and
\begin{equation}
 \label{var_est_strat}
\m{var}(\widehat{c})=\sum_{i=1}^I\frac{p_i^2\sigma_i^2}{N_i}=
\frac{1}{N}\sum_{i=1}^I\frac{p_i^2\sigma_i^2}{q_i}=
\frac{1}{N}\sum_{i=1}^I\Big(\frac{p_i\sigma_i}{q_i}\Big)^2q_i\geq
\frac{1}{N}\Big(\sum_{i=1}^I\frac{p_i\sigma_i}{q_i}q_i\Big)^2,
\end{equation}
where $\sigma_i^2=\m{var}(f(X_i))=\m{var}(f(X)|X\in A_i)$ for all $1\leq i\leq I$.

\vspace{0.1cm}
During all the sequel we consider that 
$$(H)\quad\sigma_i>0\text{ for at least one index }i.$$

\vspace{0.2cm}
The brute force Monte Carlo estimator of $\m{esp}f(X)$ is
$\frac{1}{N}\sum_{j=1}^Nf(X^j)$, with the $X^j$'s i.i.d. drawings of
$X$. Its variance is
$$\frac{1}{N}\left(\sum_{i=1}^Ip_i(\sigma_i^2+\m{esp}^2(f(X_i)))-\left(\sum_{i=1}^Ip_i\m{esp}(f(X_i))\right)^2\right)\geq \frac{1}{N}\sum_{i=1}^Ip_i\sigma_i^2.$$

For given strata the stratified estimator achieves variance reduction if
the allocations $N_i$ or equivalently the proportions $q_i$ are properly
chosen. For instance, for the so-called proportional allocation
$q_i=p_i$, $\forall i$, the variance of the stratified estimator is
equal to the previous lower bound of the variance of the brute force
Monte Carlo estimator. For the choice
$$q_i=\dfrac{p_i\sigma_i}{\sum_{j=1}^Ip_j\sigma_j}=:q_i^*,\quad
\forall\, 1\leq i\leq I,$$
the lower-bound in \eqref{var_est_strat} is attained. We speak of {\it optimal allocation}. We then have
$$\m{var}(\widehat{c})=\frac{1}{N}\Big(\sum_{i=1}^Ip_i\sigma_i\Big)^2=:\frac{\sigma_*^2}{N},$$
and no choice of the $q_i$'s can achieve a smaller variance of $\widehat{c}$. 

In general when the conditional expectations $\m{esp}(f(X)|X\in A_i)=\m{esp}(f(X_i))$ are
unknown, then so are the conditional variance $\sigma_i^2$. Therefore
optimal allocation of the drawings is not feasible at once. One can of
course estimate the conditional variances and the optimal proportions by a first Monte Carlo
algorithm and run a second Monte Carlo procedure with drawings
independent from the first one to compute the stratified estimator
corresponding to these estimated proportions. But, as suggested in
\cite{arou} in the different context of importance sampling methods, it
is a pity not to use the drawings made in the first Monte Carlo
procedure also for the final computation of the conditional expectations.

Instead of running two successive Monte Carlo procedures, we can
think to get a first estimation of the  $\sigma_i$'s, using the first
drawings of the $X_i$'s made to compute the stratified estimator. We could
then estimate the optimal allocations before making further drawings allocated in the strata according to these estimated proportions. We can next get another estimation of the $\sigma_i$'s, compute
again the allocations and so on. Our goal is thus to design and study
such an {\it adaptive stratified estimator}. The
estimator is described in Section \ref{sec_algo}. In particular, we
propose a version of the algorithm such that at each step, the
allocation of the new
drawings in the strata is not simply proportional to the current estimation of
the optimal proportions but chosen in order to minimize the variance of
the stratified estimator at the end of the step. A Central Limit Theorem for this
estimator is shown in Section \ref{sec-vit}. The asymptotic variance is
equal to the optimal variance $\sigma_*^2$ and our estimator is
asymptotically optimal. In Section \ref{snum}, we confirm the efficiency
of our algorithm by numerical experiments. We first deal with a toy
example before considering the pricing of an arithmetic average Asian
option in the Black-Scholes model.

Another stratified sampling algorithm in which the optimal proportions
and the conditional expectations are estimated using the same drawings
has been very recently proposed in \cite{cann} for quantile estimation. More precisely, for a
total number of drawings equal to $N$, the authors suggest to allocate
the $N^\gamma$ with $0<\gamma<1$ first ones proportionally to the probabilities of the
strata and then use the estimation of the optimal proportions obtained
from these first drawings to allocate the $N-N^\gamma$ remaining
ones. Their stratified estimator is also asymptotically normal with
asymptotic variance equal to the optimal one. In practice, $N$ is finite
and it is better to take
advantage of all the drawings and not only the $N^\gamma$ first ones to modify adaptively the allocation between the
strata. Our algorithm works in this spirit.


\section{The algorithm}
\label{sec_algo}

The construction of the adaptive stratified estimator relies on steps at which we estimate
the conditional variances and compute the allocations. We denote by
$N^k$ the total number of drawings made in all the strata up to the end
of step~$k$. By convention, we set $N_0=0$. In order to be able to make
one drawing in each stratum at each step we assume that $N^k-N^{k-1}\geq I$
for all $k\geq 1$.

For all $1\leq i\leq I$ we denote by $N_i^k$ the number of drawings in stratum $i$
till the end of step $k$ with convention $N_i^0=0$. The increments $M_i^k=N_i^k-N^{k-1}_i$'s are computed at the beginning of
step $k$ using the information contained in the $N^{k-1}$ first drawings. 

\vspace{0.5cm}
{\bf STEP $k\geq 1$.} 
\vspace{0.1cm}

{\it Computation of the empirical variances.}

If $k>1$, for all $1\leq i\leq I$ compute
$$\widehat{\sigma}_i^{k-1}=\sqrt{\frac{1}{N_i^{k-1}}
\Big(\sum_{j=1}^{N_i^{k-1}}(f(X_i^j))^2
-\big(\frac{1}{N_i^{k-1}}\sum_{j=1}^{N_i^{k-1}}f(X_i^j)\big)^2\Big)}.$$
If $k=1$, set $\widehat{\sigma}_i^{0}=1$ for $1\leq i\leq I$.

\vspace{0.3cm}

{\it Computation of the allocations $M_i^k=N_i^k-N^{k-1}_i$.}

We make at least one drawing in each stratum. This ensures the convergence of the estimator and of the $\widehat{\sigma}_i^k$'s (see the proof of Proposition \ref{conv} below).

That is to say we have,
\begin{equation}
\label{p1}
\forall \,1\leq i\leq I,\quad M_i^k=1+\tilde{m}_i^k, \quad\mathrm{with}\quad\tilde{m}_i^k\in\m{N},
\end{equation}
 and we now seek the
$\tilde{m}_i^k$'s. We have $\sum_{i=1}^I\tilde{m}_i^k=N^k-N^{k-1}-I$, and possibly
$\tilde{m}_i^k=0$ for some indexes. 

\vspace{0.2cm}
We present two possible ways to compute the $\tilde{m}_i^k$'s.

\vspace{0.2cm}

{\it a)}   We know that the optimal proportion of total drawings in stratum $i$ for the stratified estimator is 
$q_i^*=\frac{p_i\sigma_i}{\sum_{j=1}^Ip_j\sigma_j}$, so we may want to choose the vector
$(\tilde{m}_1^k,\hdots,\tilde{m}_I^k)\in \m{N}^I$ close to
$({m}_1^k,\hdots,{m}_I^k)\in \m{R}_+^I$ defined by $$m_i^k=\dfrac{p_i\widehat{\sigma}_i^{k-1}}{\sum_{j=1}^Ip_j\widehat{\sigma}_j^{k-1}}
(N^k-N^{k-1}-I)\mbox{ for }1\leq i\leq I.$$
This can be achieved by setting 
$$\tilde{m}_i^k=\lfloor m_1^k+\ldots+m_i^k\rfloor-\lfloor m_1^k+\ldots+m_{i-1}^k\rfloor,$$
with the convention that the second term is zero for $i=1$. This
systematic sampling procedure ensures that
$\sum_{i=1}^I\tilde{m}^k_i=N^k-N^{k-1}-I$ and $m^k_i-1<\tilde{m}^k_i< m_i^k+1$
for all $1\leq i\leq I$.
In case $\widehat{\sigma}_i^{k-1}=0$ for all $1\leq i\leq
I$, the above definition of $m_i^{k}$ does not make sense and we set $m_i^k=p_i(N^k-N^{k-1}-I)$ for $1\leq i\leq
I$ before applying the systematic sampling procedure. Note that thanks to $(H)$ and the convergence of the 
$\widehat{\sigma}_i^k$ (see Proposition \ref{conv} below), this
asymptotically will never be the case.

\vspace{0.2cm}

{\it b)} 
In case $\widehat{\sigma}_i^{k-1}=0$ for all $1\leq i\leq
I$, we do as before. Otherwise, we may think to the expression of the variance of the stratified estimator 
with allocation $N_i$ for all $i$, which is given by \eqref{var_est_strat},
and find $(m_1^k,\ldots,m_I^k)\in\m{R}_+^I$ that minimizes
$$\sum_{i=1}^I\dfrac{p_i^2(\widehat{\sigma}_i^{k-1})^2}{N_i^{k-1}+1+m_i^k},$$
under the constraint $\sum_{i=1}^Im_i^k=N^k-N^{k-1}-I$.

This can be done in the following manner (see in the Appendix Proposition~\ref{prop-opt}):

\vspace{0.3cm}

\vspace{0.1cm}
For the indexes $i$ such that
$\widehat{\sigma}_i^{k-1}=0$, we set $m_i^k=0$.

We denote $I^k$ the number of indexes such that $\widehat{\sigma}_i^{k-1}>0$. We renumber
the corresponding strata from $1$ to $I^k$. We now find 
$(m_1^k,\ldots,m_{I^k}^k)\in\m{R}_+^{I^k}$ that minimizes
$\sum_{i=1}^{I^k} \frac{p_i^2(\widehat{\sigma}_i^{k-1})^2}{N_i^{k-1}+1+m_i^k}$,
under the constraint $\sum_{i=1}^{I^k}m_i^k=N^k-N^{k-1}-I$, by applying 
the three following points:
\vspace{0.1cm}

i) Compute the quantities $\frac{N_i^{k-1}+1}{p_i\widehat{\sigma}_i^{k-1}}$ and sort them in decreasing order. Denote by $\frac{N_{(i)}^{k-1}+1}{p_{(i)}\widehat{\sigma}_{(i)}^{k-1}}$ the ordered quantities.

 ii) For $i=1,\ldots,I^k$ compute the quantities
$$\frac{\displaystyle N^k-N^{k-1}-I+\sum_{j=i+1}^{I^k}(N^{k-1}_{(j)}+1)}{\displaystyle \sum_{j=i+1}^{I^k}p_{(j)}\widehat{\sigma}^{k-1}_{(j)}}.$$
Denote by $i^*$ the last $i$ such that
$$\dfrac{N_{(i)}^{k-1}+1}{p_{(i)}\widehat{\sigma}_{(i)}^{k-1}}\geq \frac{\displaystyle N^k-N^{k-1}-I+\sum_{j=i+1}^{I^k}(N^{k-1}_{(j)}+1)}{\displaystyle \sum_{j=i+1}^{I^k}p_{(j)}\widehat{\sigma}_{(j)}^{k-1}}.$$

If this inequality is false for all $i$, then by convention $i^*=0$. 

iii) Then for $i\leq i^*$ set $m^k_{(i)}=0$ and for $i>i^*$,
$$m^k_{(i)}=p_{(i)}\widehat{\sigma}_{(i)}^{k-1}.\frac{\displaystyle N^k-N^{k-1}-I+\sum_{j=i^*+1}^{I^k}(N^{k-1}_{(j)}+1)}
{\displaystyle\sum_{j=i^*+1}^{I^k}p_{(j)}
\widehat{\sigma}_{(j)}^{k-1}}-N^{k-1}_{(i)}-1.$$
This quantity is non-negative according to the proof of Proposition \ref{prop-opt}.

\vspace{0.2cm}
We then build $(m_1^k,\ldots,m_I^k)$ by reincluding the $I-I^k$ zero valued $m_i^k$'s
and using the initial indexation.
Finally we deduce $(\tilde{m}_1^k,\ldots,\tilde{m}_I^k)\in\m{N}^I$ by
the systematic sampling procedure described in {\it a)}.

\vspace{0.2cm}


{\it Drawings of the $X_i$'s.}
 Draw $M_i^k$ i.i.d. realizations of $X_i$ in each stratum $i$ and set
 $N_i^k=N_i^{k-1}+M_i^k$.

\vspace{0.3cm}
{\it Computation of the estimator} 

Compute
\begin{equation}
\label{val_est}
\hat{c}^k:=\sum_{i=1}^I\frac{p_i}{N_i^k}\sum_{j=1}^{N_i^k}f(X_i^j).
\end{equation}

\vspace{0.3cm}
Square integrability of $f(X)$ is not necessary in order to ensure that
the estimator $\widehat{c}^k$ is strongly consistent. Indeed thanks to
\eqref{p1}, we have $N_i^k\to\infty$ as $k\to\infty$ and the strong law
of large numbers ensures the following Proposition.

\begin{prop}
\label{conv}
If $\m{esp}|f(X)|<+\infty$, then 
$$\widehat{c}^k\xrightarrow[k\to\infty]{}c\quad\mathrm{a.s..}$$
If moreover, $\m{esp}(f^2(X))<+\infty$, then a.s., 
$$\forall 1\leq i\leq
I,\;\widehat{\sigma}_i^k\xrightarrow[k\to\infty]{}\sigma_i\;\;\mbox{
  and }\;\;\sum_{i=1}^I p_i\widehat{\sigma}_i^k \xrightarrow[k\to\infty]{}\sigma_*.$$
\end{prop}


\section{Rate of convergence}
\label{sec-vit}

In this section we prove the following result.

\begin{thm}
\label{thm-vit-conv} 
Assume $(H)$, $\m{esp}(f^2(X))<+\infty$ and $k/N^k\to 0$ as $k\to\infty$.
Then, using either procedure {\it a)} or procedure {\it b)} for the computation of
allocations, one has
$$\sqrt{N^k}\big(\hat{c}^k-c\big)\xrightarrow[k\to\infty]{\mathrm{in law}}
\m{nor}(0,\sigma_*^2).$$
\end{thm}
With Proposition \ref{conv}, one deduces that $\frac{\sqrt{N^k}}{\sum_{i=1}^I p_i\widehat{\sigma}_i^k}\big(\hat{c}^k-c\big)\xrightarrow[k\to\infty]{\mathrm{in law}}
\m{nor}(0,1)$, which enables the easy construction of confidence intervals.
The theorem is a direct consequence of the two following
propositions.

\begin{prop}
\label{prop-vit-conv} 
If $\m{esp}(f^2(X))<+\infty$ and \begin{equation}
\label{nikq}
\forall 1\leq i\leq I,\quad \frac{N_i^k}{N^k}\xrightarrow[k\to\infty]{}q_i^*\;\;\mathrm{a.s.},
\end{equation}
then 
$$\sqrt{N^k}\big(\hat{c}^k-c\big)\xrightarrow[k\to\infty]{\mathrm{in law}}
\m{nor}(0,\sigma_*^2).$$
\end{prop}

\begin{prop}
\label{convnink}
Under the assumptions of Theorem \ref{thm-vit-conv}, using either procedure {\it a)} or procedure {\it b)} for the computation of
allocations,
\eqref{nikq} holds.

\end{prop}

We prove Proposition \ref{prop-vit-conv} and \ref{convnink} in the following subsections.

\subsection{Proof of Proposition \ref{prop-vit-conv}}

The main tool of the proof of this proposition will be a CLT for martingales that we recall below.

\begin{thm}[Central Limit Theorem]
\label{clt} 
Let $(\mu_n)_{n\in\m{N}}$ be a square-integrable $(\m{F}_n)_{n\in\m{N}}$-vector martingale. Suppose that for a deterministic sequence $(\gamma_n)$ increasing to $+\infty$ we have,

i) 
$$\dfrac{\langle\mu\rangle_n}{\gamma_n}\xrightarrow[n\to\infty]{\m{P}}\Gamma.$$

ii) The Lindeberg condition is satisfied, i.e. for all $\varepsilon>0$
$$\dfrac{1}{\gamma_n}\sum_{k=1}^n
\m{esp}\Big[||\mu_k-\mu_{k-1}||^2
\m{ind}_{\{||\mu_k-\mu_{k-1}||\geq\varepsilon\sqrt{\gamma_n}\}}|\m{F}_{k-1}\Big]
\xrightarrow[n\to\infty]{\m{P}}0.
$$

Then 
$$\dfrac{\mu_n}{\sqrt{\gamma_n}}\xrightarrow[n\to\infty]{\mathrm{in law}}\m{nor}(0,\Gamma).$$
\end{thm}

\vspace{0.6cm}
As we can write
$$\sqrt{N^k}\big(\hat{c}^k-c\big)=
\left(\begin{array}{c}
          p_1\frac{N^k}{N_1^k}\\
\vdots\\
p_I\frac{N^k}{N_I^k}\\
         \end{array}\right).\frac{\displaystyle 1}{\displaystyle\sqrt{N^k}}
\left(\begin{array}{c}
          \sum_{j=1}^{N_1^k}(f(X_1^j)-\m{esp}f(X_1))\\
\vdots\\
\sum_{j=1}^{N_I^k}(f(X_I^j)-\m{esp}f(X_I))\\
         \end{array}\right),
$$
we could think to set 
$\mu_k:=\Big(\sum_{j=1}^{N_1^k}(f(X_1^j)-\m{esp}f(X_1)),\ldots,\sum_{j=1}^{N_I^k}(f(X_I^j)-\m{esp}f(X_I))\Big)'$ and try to use Theorem \ref{clt}. Indeed if we define the filtration 
$(\mathcal{G}_k)_{k\in\m{N}}$ by 
$\mathcal{G}_k=\sigma(\m{ind}_{j\leq N_i^k}X_i^j,\,1\leq i\leq I,\,\, 1\leq j)$, it can be shown that
$(\mu_k)$ is a $(\mathcal{G}_k)$-martingale. This is thanks to the fact that the $N_i^k$'s
are $\mathcal{G}_{k-1}$-measurable. Then easy computations show that
$$\frac{1}{N^k}\langle\mu\rangle_k=
\mathrm{diag}\Big(\big(\frac{N_1^k}{N^k}\sigma_1^2,\ldots,\frac{N_I^k}{N^k}\sigma_I^2\big)\Big)$$
where $\mathrm{diag}(\mathbf{v})$ denotes the diagonal matrix with vector $\mathbf{v}$ on the diagonal. Thanks to \eqref{nikq} we thus have
$$\frac{1}{N^k}\langle\mu\rangle_k\xrightarrow[k\to\infty]{\mathrm{a.s.}}
\mathrm{diag}\Big(\big(q_1^*\sigma_1^2,\ldots,q_I^*\sigma_I^2\big)\Big),$$
and a use of Theorem \ref{clt} and Slutsky's theorem could lead to the desired result.

\vspace{0.1cm}
The trouble is that Lindeberg's condition cannot be verified in this context, and we will not be able to apply Theorem \ref{clt}. Indeed the quantity $||\mu_k-\mu_{k-1}||^2$ involves $N^k-N^{k-1}$
random variables of the type $X_i$ and we cannot control it without making some growth assumption on $N^k-N^{k-1}$.

\vspace{0.3cm}
In order to handle the problem, we are going to introduce a microscopic scale. From the sequence of estimators 
$(\hat{c}^k)$ we will build a sequence 
$(\tilde{c}^n)$ of estimators of $c$, such that $\hat{c}^k=\tilde{c}^{N^k}$, and for which we will show a CLT using
Theorem \ref{clt}. It will be possible because it involves a new martingale
$(\mu_n)$ such that $\mu_n-\mu_{n-1}$ is equal to a vector the only non zero coordinate of which is {\it one} random variable $f(X_i^j)$. Then the Lindeberg condition will be easily verified, but this time we will have to work a little more to check the bracket condition. As the sequence $(\hat{c}^k)$
is a subsequence of $(\tilde{c}^n)$, Proposition \ref{prop-vit-conv}
will follow. This is done in the following way.

\vspace{0.1cm}
Let $n\in\m{N}^*$. In the setting of the Algorithm of Section
 \ref{sec_algo} let $k\in\m{N}$ such that
$N^{k-1}< n\leq N^k$. Given the allocations $(N_i^{l})_{i=1}^I$, for $0\leq l\leq k$,
we define for each $1\leq i\leq I$ a quantity $\nu_i^n$ with the inductive rule below. Each
$\nu_i^n$ is the number of drawings in the $i$-th strata among the first $n$ drawings and we have $\sum_{i=1}^I\nu_i^n=n$.
  We then define 

\begin{equation*}
 \widetilde{c}^n:=\sum_{i=1}^I\frac{p_i}{\nu_i^n}\sum_{j=1}^{\nu_i^n}f(X_i^j).
\end{equation*}

\newpage
\begin{center}
 {\bf Rule for the $\nu_i^n$'s}
\end{center}

 For $n=0$, $\nu_i^n=0$, for all $1\leq i\leq I$.

\vspace{0.2cm}

\begin{center}
 \begin{enumerate}
 \item For $k>0$ set $r_i^k:=\frac{N_i^k-N_i^{k-1}}{N^k-N^{k-1}}$ for $1\leq i\leq I$.
\item For $N^{k-1}<n\leq N^k$, and given the $\nu_i^{n-1}$'s find
$$i_n=\argmax_{1\leq i\leq I}
\Big(r_i^k-\frac{\nu_i^{n-1}-N_i^{k-1}}{n-N^{k-1}}\Big).$$
If several $i$ realize the maximum choose $i_n$ to be the one for which $r_i^k$ is the greatest. If there are still ex aequo's choose the greatest $i$.
\item Set $\nu_{i_n}^n=\nu_{i_n}^{n-1}+1$, and $\nu_i^n=\nu_i^{n-1}$ if $i\neq i_n$.
\end{enumerate}
\end{center}

There is always an index $i$ for which $r_i^k-\frac{\nu_i^{n-1}-N_i^{k-1}}{n-N^{k-1}}>0$, since
$$\sum_{i=1}^I\frac{\nu_i^{n-1}-N_i^{k-1}}{n-N^{k-1}}=\frac{n-1-N^{k-1}}{n-N^{k-1}}
<1=\sum_{i=1}^Ir_i^k.$$

 Moreover, for the first  $n\in\{N^{k-1}+1,\hdots,N^k\}$ such that
$\nu_i^{n-1}=N_i^k$ in the $i$-th strata, $r_i^k-\frac{\nu_i^{n-1}-N_i^{k-1}}{n-N^{k-1}}\leq 0$ and
$\nu_i^{n'}=\nu_i^n=N_i^k$ for $n\leq n'\leq N^k$.

This implies that
$$\nu_i^{N^k}=N_i^k,\;\;\forall 1\leq i\leq I,\;\;\forall k\in\m{N},$$
and as a consequence,
\begin{equation}
\label{subseq}
\hat{c}^k=\tilde{c}^{N^k}.
\end{equation}
Therefore Proposition \ref{prop-vit-conv} is an easy consequence of the
following one.

\begin{prop}
\label{thm_main}
Under the assumptions of Proposition \ref{prop-vit-conv},
$$\sqrt{n}\big(\tilde{c}^n-c\big)\xrightarrow[n\to\infty]{\mathrm{in law}}
\m{nor}(0,\sigma_*^2).$$

\end{prop}

\vspace{0.3cm}

In the proof of Proposition \ref{thm_main}, to verify the bracket condition of Theorem \ref{clt}, we will need the following result. 

\begin{lem}
 \label{propnu}
When \eqref{nikq} holds, then
$$\forall 1\leq i\leq I,\quad\frac{\nu_i^n}{n}\xrightarrow[n\to\infty]{}q_i^*\;\;\mathrm{a.s.}$$
\end{lem}

\begin{proof}
Let be $1\leq i\leq I$. During the sequel, for $x\in\m{R}_+^*$ or $n\in\m{N}^*$, the integer~$k$ is implicitely such that  $N^{k-1}<x,n\leq N^k$.

We notice that for any $n\in\m{N}^*$
$$\frac{\nu_i^n}{n}=\frac{n-N^{k-1}}{n}.\frac{\nu_i^n-N_i^{k-1}}{n-N^{k-1}}
+\frac{N^{k-1}}{n}.\frac{N_i^{k-1}}{N^{k-1}},
$$
and define for $x\in\m{R}_+^*$,
$$f(x):=\frac{x-N^{k-1}}{x}.\frac{N_i^k-N_i^{k-1}}{N^k-N^{k-1}}
+\frac{N^{k-1}}{x}.\frac{N_i^{k-1}}{N^{k-1}}.
$$

We will see that,
as $n$ tends to infinity, $f(n)$ tends to $q_i^*$ and  $f(n)-\frac{\nu_i^n}{n}$ tends to zero.

Computing the derivative of $f$ on any interval $(N^{k-1},N^k]$ we find that this function is monotonic on it. Besides $f(N^{k-1})=\frac{N_i^{k-1}}{N^{k-1}}$ and $f(N^k)=\frac{N_i^k}{N^k}$.
So if $\frac{N_i^k}{N^k}$ tends to $q_i^*$ as $k$ tends to infinity, we can conclude that 
\begin{equation}
 \label{ftq}
f(n)\xrightarrow[n\to\infty]{}q_i^*.
\end{equation}

As $r_i^k=\frac{N_i^k-N_i^{k-1}}{N^k-N^{k-1}}$ we now write
$$\dfrac{\nu_i^n}{n}-f(n)=\dfrac{n-N^{k-1}}{n}
\Big(\frac{\nu_i^n-N_i^{k-1}}{n-N^{k-1}}-r_i^k\Big).$$
We conclude the proof by checking that 
\begin{equation}
 \label{micro}
r_i^k-\dfrac{I-1}{n-N^{k-1}}<\dfrac{\nu_i^n-N_i^{k-1}}{n-N^{k-1}}<r_i^k+\dfrac{1}{n-N^{k-1}}.
\end{equation}
Indeed, this inequality implies
$$-\dfrac{I-1}{n}<\dfrac{\nu_i^n}{n}-f(n)<\dfrac{1}{n},$$
which combined with \eqref{ftq} gives the desired conclusion.
We first show
\begin{equation}
\label{maj}
\dfrac{\nu_i^n-N_i^{k-1}}{n-N^{k-1}}<r_i^k+\dfrac{1}{n-N^{k-1}}.
\end{equation}
We distinguish two cases.
Either 
 $\nu_i^{n'}=N_i^{k-1}$ for all $N^{k-1}<n'\leq n$, that is to say no
drawing at all is made in stratum $i$ between $N^{k-1}$ and $n$, then
\eqref{maj} is trivially verified.

Either some drawing is made between $N^{k-1}$ and $n$. Let us denote by $n'$
the index of the last one, i.e. we have
$\nu_i^n=\nu_i^{n'}=\nu_i^{n'-1}+1$.
As a drawing is made at $n'$ we have
$\frac{\nu_i^{n'-1}-N_i^{k-1}}{n'-N^{k-1}}<r_i^k$.

We thus have,
$$\dfrac{\nu_i^{n'-1}-N_i^{k-1}}{n-N^{k-1}}\leq
\dfrac{\nu_i^{n'-1}-N_i^{k-1}}{n'-N^{k-1}}<r_i^k$$
and
$$\dfrac{\nu_i^{n}-N_i^{k-1}}{n-N^{k-1}}=\dfrac{\nu_i^{n'-1}+1-N_i^{k-1}}{n-N^{k-1}},$$
and thus we have again \eqref{maj}.

Using now the fact that
$1=\sum_{i=1}^Ir_i^k=\sum_{i=1}^I\frac{\nu_i^n-N_i^{k-1}}{n-N^{k-1}}$ we
get
$$\dfrac{\nu_i^n-N_i^{k-1}}{n-N^{k-1}}=r_i^k+\sum_{i\neq
  j}\Big(r_j^k-\dfrac{\nu_j^n-N_i^{k-1}}{n-N^{k-1}}\Big)$$

Using this and \eqref{maj} we get \eqref{micro}.
\end{proof}

\begin{proof}[Proof of Proposition \ref{thm_main}.] For $n\geq N^1$,
  $\nu_i^n\geq 1$ for all $1\leq i\leq I$ and we can write
\begin{equation}
\label{slu1} 
\sqrt{n}\big(\tilde{c}^n-c\big)=
\left(\begin{array}{c}
          p_1\frac{n}{\nu_1^n}\\
\vdots\\
p_I\frac{n}{\nu_I^n}\\
         \end{array}\right).\frac{\displaystyle 1}{\displaystyle\sqrt{n}}\mu_n,
\end{equation}
with
$$\mu_n=\left(\begin{array}{c}
          \sum_{j=1}^{\nu_1^n}(f(X_1^j)-\m{esp}f(X_1))\\
\vdots\\
\sum_{j=1}^{\nu_I^n}(f(X_I^j)-\m{esp}f(X_I))\\
         \end{array}\right).
$$
Note that if $\sigma_i=0$ for a stratum $i$, then $q_i^*=0$ and by Lemma
\ref{propnu},
$\frac{n}{\nu_i^n}\xrightarrow[n\to\infty]{\mathrm{a.s.}}+\infty$ which
may cause some trouble in the convergence analysis. In compensation, $\sigma_i=0$ means 
that $f(X_i)-\m{esp}f(X_i)=0$ a.s. Thus the component $\mu_n^i$ of $\mu_n$ 
makes no contribution in $\tilde{c}^n-c$. So we might rewrite \eqref{slu1} 
with $\mu_n$ a vector of size less than $I$, whose components correspond only to indexes $i$ with $\sigma_i> 0$. For the seek of simplicity we keep the size $I$ and consider that $\sigma_i> 0$ for all $1\leq i\leq I$.  

If we define $\m{F}_n:=\sigma(\m{ind}_{j\leq \nu_i^n}X_i^j,\,1\leq i\leq I,\,1\leq j)$,
then
 $(\mu_n)_{n\geq 0}$ is obviously a $(\m{F}_n)$-martingale. Indeed, for $n\in\m{N}^*$ let  $k\in\m{N}^*$ such that $N^{k-1}<n\leq N^k$. For 
$1\leq i\leq I$ the variables $N_i^{k-1}$ and $N_i^k$  are respectively $\m{F}_{N^{k-2}}$ and $\m{F}_{N^{k-1}}$-measurable (Step $k>1$ in the Algorithm). As for each $1\leq i\leq I$ the quantity $\nu_i^n$ depends on the $N_i^{k-1}$'s and the
 $N_i^k$'s, it is $\m{F}_{N^{k-1}}$-measurable. Thus $\mu_n$ is $\m{F}_n$-measurable and
easy computations show that $\m{esp}[\mu_{n+1}|\m{F}_n]=\mu_n$.

We wish to use Theorem \ref{clt} with $\gamma_n=n$. 
We will denote by $\mathrm{diag}(a_i)$ the $I\times I$ matrix having null coefficients except the $i$-th diagonal term with value~$a_i$.

We first verify the Lindeberg condition. We have, using the sequence $(i_n)$ defined in the rule for the $\nu_i^n$'s,
$$\begin{array}{ll}
& \frac{\displaystyle 1}{\displaystyle n}\sum_{l=1}^n
\m{esp}\big[||\mu_l-\mu_{l-1}||^2
\m{ind}_{\{||\mu_l-\mu_{l-1}||>\varepsilon\sqrt{n}\}}|\m{F}_{l-1} \big]\\
\\
=&\frac{\displaystyle 1}{\displaystyle n}\sum_{l=1}^n
\m{esp}\big[|f(X_{i_l}^{\nu_{i_l}^l})-\m{esp}f(X_{i_l})|^2
\m{ind}_{\{|f(X_{i_l}^{\nu_{i_l}^l})-\m{esp}f(X_{i_l})|>\varepsilon\sqrt{n}\}}|\m{F}_{l-1} \big]\\
\\
\leq&\frac{\displaystyle 1}{\displaystyle n}\sum_{l=1}^n
\sup_{1\leq i\leq I}\m{esp}\big[|f(X_i)-\m{esp}f(X_{i})|^2\m{ind}_{\{|f(X_i)-\m{esp}f(X_{i})|>\varepsilon\sqrt{n}\}}\big]\\
\\
=&\sup_{1\leq i\leq I}\m{esp}\big[|f(X_i)-\m{esp}f(X_{i})|^2\m{ind}_{\{|f(X_i)-\m{esp}f(X_{i})|>\varepsilon\sqrt{n}\}}\big].\\
\end{array}
$$

As
$$\sup_{1\leq i\leq I}\m{esp}\big[|f(X_i)-\m{esp}f(X_{i})|^2\m{ind}_{\{|f(X_i)-\m{esp}f(X_{i})|>\varepsilon\sqrt{n}\}}\big]
\xrightarrow[n\to\infty]{}0,$$
the Lindeberg condition is proven.

We now turn to the bracket condition. We have,
$$\begin{array}{lll}
 \langle\mu\rangle_n&=&\sum_{k=1}^n
\m{esp}\big[(\mu_k-\mu_{k-1})(\mu_k-\mu_{k-1})'|\m{F}_{k-1}\big]\\
\\
&=&\sum_{k=1}^n\mathrm{diag}\Big(\m{esp}
\big[\,\big|f(X_{i_k}^{\nu_{i_k}^k})-\m{esp}f(X_{i_k})\big|^2\,\big]\Big)\\
\\
&=&\sum_{k=1}^n\mathrm{diag}\Big(\sigma_{i_k}^2\Big).\\
\end{array}
 $$

Thus,  we have
$$\frac{\langle\mu\rangle_n}{n}=\mathrm{diag}
\big(\,(\frac{\nu_1^n}{n}\sigma_1^2,\ldots,\frac{\nu_I^n}{n}\sigma_I^2)\,\big)
\xrightarrow[n\to\infty]{}\mathrm{diag}
\big(\,(q_1^*\sigma_1^2,\ldots,q_I^*\sigma_I^2)\,\big)\;\;\mathrm{a.s.},$$
where we have used Lemma \ref{propnu}.

Theorem \ref{clt} implies that
\begin{equation}
\label{eqTCL}
 \frac{\mu_n}{\sqrt{n}}\xrightarrow[n\to\infty]{\mathrm{in law}}
\m{nor}\Big(0,\mathrm{diag}
\big(\,(q_1^*\sigma_1^2,\ldots,q_I^*\sigma_I^2)\,\big)\Big).
\end{equation}

Using again Lemma \ref{propnu} we have
\begin{equation}
 \label{sluvec}
(p_1\frac{n}{\nu_1^n},\ldots,p_I\frac{n}{\nu_I^n})
\xrightarrow[n\to\infty]{}(\frac{p_1}{q_1^*},\ldots,\frac{p_I}{q_I^*})\;\;\mathrm{a.s.}
\end{equation}

 Using finally Slutsky's theorem, \eqref{slu1}, \eqref{eqTCL} and \eqref{sluvec}, we get,
$$ \sqrt{n}\big(\tilde{c}^n-c\big)
\xrightarrow[n\to\infty]{\mathrm{in law}}
\m{nor}\big(0,\sigma_*^2).$$
\end{proof}

\subsection{Proof of Proposition  \ref{convnink}}
Thanks to $(H)$ and Proposition \ref{conv} there exists $K\in\m{N}$ s.t. for all
$k\geq K$ we have $\sum_{i=1}^Ip_i\widehat{\sigma}_i^k>0$.
The proportions
$(\rho_i^k=\frac{p_i\widehat{\sigma}_i^k}{\sum_{j=1}^Ip_j\widehat{\sigma}_j^k})_i$
are well defined for all $k\geq K$ and play an important role in both
allocation rules {\it a)} and {\it b)}. Proposition \ref{conv} implies convergence of
$\rho_i^k$ as $k\rightarrow+\infty$.
\begin{lem}
\label{lem-rhob}
Under the assumptions of Theorem \ref{thm-vit-conv}, $$\forall  1\leq i\leq I,\quad \rho_i^k\xrightarrow[k\to\infty]{}q_i^*\;\;\mathrm{a.s.}$$
\end{lem}
\begin{proof}[Proof of Proposition \ref{convnink} for allocation rule {\it a)}]
 Let be $1\leq i\leq I$. We have $\frac{N_i^k}{N^k}=\frac{k+\sum_{l=1}^k\tilde{m}_i^l}{N^k}$. Using the fact
that $m_i^l-1<\tilde{m}_i^l<m_i^l+1$ we can write
$$\dfrac{\sum_{l=1}^km_i^l}{N^k}\leq \dfrac{N_i^k}{N^k}\leq
\dfrac{2k}{N^k}+\dfrac{\sum_{l=1}^km_i^l}{N^k}.$$

We will show that $\frac{\sum_{l=1}^km_i^l}{N^k}\to q_i^*$, and, as $\frac{k}{N^k}\to 0$, will get the desired result.

For $k\geq K+1$, we have
\begin{align*}
   &\dfrac{\sum_{l=1}^km_i^l}{N^k}=\dfrac{\sum_{l=1}^{K} m_i^l}{N^k}+\dfrac{\sum_{l=K+1}^k\rho_i^l(N^l-N^{l-1}-I)}{N^k}\\&=\dfrac{\sum_{l=1}^{K} m_i^l}{N^k}+\frac{N^k-N^{K}}{N^k}\times\frac{1}{N^k-N^{K}}\sum_{n=N^{K}+1}^{N^k}\tilde{\rho}_i^n-\frac{I(k-K)}{N^k}\times\frac{1}{k-K}\sum_{l=K}^k\rho_i^l
\end{align*}
where the sequence $(\tilde{\rho}_i^n)$ defined by
$\tilde{\rho}_i^n=\rho_i^l$ for $N^{l-1}<n\leq N^l$ converges to $q_i^*$
as $n$ tends to infinity. The Cesaro
means which appear as factors in the second and third terms of the
r.h.s. both converge a.s. to $q_i^*$. One easily deduce that the first,
second and third terms respectively converge to $0$, $q_i^*$ and $0$.
\end{proof}
\begin{proof}[Proof of Proposition \ref{convnink} for allocation rule {\it b)}]

\vspace{0.3cm}

There may be some strata of zero variance. We denote by $I'$
 ($I'\leq I$) the number of strata of non zero variance.

For a stratum $i$ of zero variance the only drawing made at each step will be the one forced by \eqref{p1}. Indeed $\widehat{\sigma}_i^k=0$ for all $k$ in this case. 
Thus $N_i^k=k$ for all the strata of zero variance and since
$\frac{k}{N^k}\rightarrow 0$, we get the desired result  for them (note that of course
$q_i^*=0$ in this case).

\vspace{0.4cm}

We now work on the $I'$ strata such that $\sigma_i>0$. We renumber these strata
from $1$ to $I'$. Let now $K'$ be such that $\widehat{\sigma}_i^{k}>0$
for all $k\geq K'$, and all $1\leq i\leq I'$. For $k\geq K'$, the integer $I^{k+1}$ at step $k+1$ in procedure
{\it b)} is equal to $I'$. 

\vspace{0.3cm}

{\it Step 1.} We will firstly show that
\begin{equation}
\label{leqmax}
\forall k\geq K',\,\forall 1\leq i\leq I'\qquad\frac{N_i^{k+1}}{N^{k+1}}\leq
\frac{N_i^k+1}{N^{k+1}}\vee\big(\rho_i^k+\frac{1}{N^{k+1}}\big).
\end{equation}
Let $k\geq K'$. At step $k+1$ we denote by $(.)_k$ the ordered index in
Point i) of procedure~{\it b)} and by $i_k^*$ the index $i^*$ in Point
ii). We also set $n_i^{k+1}=N_i^k+1+m_i^{k+1}$. By Point
iii), for $i>i^*_k$, 
\begin{equation}
 \label{ts?aux}
\begin{array}{lll}
\dfrac{n^{k+1}_{(i)_k}}{p_{(i)_k}\widehat{\sigma}^{k}_{(i)_k}}
=\dfrac{m^{k+1}_{(i)_k}+N^{k}_{(i)_k}+1}{p_{(i)_k}\widehat{\sigma}^{k}_{(i)_k}}&=&
\dfrac{N^{k+1}-N^k-I+\sum_{j=i_k^*+1}^{I'}(N^{k}_{(j)_k}+1)}
{\sum_{j=i_k^*+1}^{I'}p_{(j)_k}\widehat{\sigma}^{k}_{(j)_k}}\\
\end{array}
\end{equation}

{\it Case 1:} $i_k^*=0$. Then, in addition to the drawing forced by \eqref{p1}, there are some drawings at step $k+1$ in stratum $(1)_k$, and consequently in all the strata. 
 Thus \eqref{ts?aux} leads to
$$n_{i}^{k+1}=\rho_i^k
\left(N^{k+1}-N^k-I+I'+\sum_{j=1}^{I'}N_{j}^k\right),\;\;\forall 1\leq i\leq I'.$$
But $N^k=\sum_{j=1}^{I'}N_{j}^k+k(I-I')$ and,
following the systematic sampling procedure, we have
\begin{equation}
 \label{niNi1} 
N_i^{k+1}< n_i^{k+1}+1,\quad \forall 1\leq i\leq I'.
\end{equation}

Thus, in this case,
$$\dfrac{N_{i}^{k+1}}{N^{k+1}}\leq \rho_i^k+\dfrac{1}{N^{k+1}},\quad\forall 1\leq i\leq I'.$$

\vspace{0.2cm}

{\it Case 2:} $i_k^*>0$. 
If $i\leq i_k^*$,
$N_{(i)_k}^{k+1}=N_{(i)_k}^{k}+1$
and \eqref{leqmax} holds.

If $i>i_k^*$, then \eqref{ts?aux} leads to
$$\frac{n_{(i)_k}^{k+1}}{N^{k+1}}
=\rho_{(i)_k}^k
\dfrac{N^{k+1}-N^k-I+\sum_{j=i_k^*+1}^{I'}(N^{k}_{(j)_k}+1)}
{N^{k+1}\sum_{j=i_k^*+1}^{I'}\rho_{(j)_k}^k}.$$
Using \eqref{niNi1}, it is enough
to check that
\begin{equation}
\label{prodcomp}
 \dfrac{N^{k+1}-N^k-I+\sum_{j=i_k^*+1}^{I'}(N^{k}_{(j)_k}+1)}
{N^{k+1}\sum_{j=i_k^*+1}^{I'}\rho_{(j)_k}^k}\leq 1
\end{equation}
in order to deduce that \eqref{leqmax} also holds for $i>i_k^*$.
 
If $\frac{N_{(i_k^*)_k}^k+1}{N^{k+1}\rho_{(i_k^*)_k}^k}\leq 1$, then inequality \eqref{prodcomp} holds by the definition of $i_k^*$.

If $\frac{N_{(i_k^*)_k}^k+1}{N^{k+1}\rho_{(i_k^*)_k}^k}>1$ we have
$\frac{N_{(i)_k}^k+1}{N^{k+1}\rho_{(i)_k}^k}>1$, $\forall i\leq i_k^*$
and thus
$$\sum_{j=1}^{i_k^*}(N_{(j)_k}^k+1)>N^{k+1}\sum_{j=1}^{i_k^*}\rho_{(j)_k}^k.$$
This inequality also writes
$$N^k-k(I-I')+I'-\sum_{j=i_k^*+1}^{I'}(N_{(j)_k}^k+1)
>N^{k+1}\big(1-\sum_{j=i_k^*+1}^{I'}\rho_{(j)_k}^k\big),$$
and \eqref{prodcomp} follows.

\vspace{0.4cm}

{\it Step 2.} Let $1\leq i\leq I'$. We set $\bar{n}_i^k:=N_i^k-k$ (this the number of drawings in stratum $i$ that 
have not been forced by \eqref{p1}).

Using \eqref{leqmax} we have 
$$\forall k\geq K',\quad\frac{N_i^{k+1}-(k+1)}{N^{k+1}}\leq \frac{N_i^k+1-(k+1)}{N^{k+1}}\vee
\big(\rho_i^k-\frac{k}{N^{k+1}}\big),$$
and thus
$$\forall k\geq K',\quad\frac{\bar{n}_i^{k+1}}{N^{k+1}}\leq \frac{\bar{n}_i^k}{N^{k+1}}\vee
\big(\rho_i^k-\frac{k}{N^{k+1}}\big).$$

Let $\varepsilon>0$. Thanks to Lemma \ref{lem-rhob}, there exists $k_0\geq K'$ s.t. for all $k\geq k_0$,  $\rho_i^k-\frac{k}{N^{k+1}}\leq q_i^*+\varepsilon$.
Thus 
\begin{equation}
 \label{aveceps}
\forall k\geq k_0,\quad\frac{\bar{n}_i^{k+1}}{N^{k+1}}\leq \frac{\bar{n}_i^k}{N^{k+1}}\vee
\big(q_i^*+\varepsilon).
\end{equation}

By induction
$$\forall k\geq k_0, \frac{\bar{n}_i^k}{N^k}\leq\frac{\bar{n}_i^{k_0}}{N^k}\vee(q_i^*+\varepsilon).
$$
Indeed suppose $\frac{\bar{n}_i^k}{N^k}\leq\frac{\bar{n}_i^{k_0}}{N^k}\vee(q_i^*+\varepsilon)$. If $\frac{\bar{n}_i^k}{N^k}\leq q_i^*+\varepsilon$ then $\frac{\bar{n}_i^k}{N^{k+1}}\leq q_i^*+\varepsilon$ and using \eqref{aveceps} we get
$\frac{\bar{n}_i^{k+1}}{N^{k+1}}\leq q_i^*+\varepsilon$.  Otherwise $\bar{n}_i^k=\bar{n}_i^{k_0}$ and using \eqref{aveceps} we are done.

\vspace{0.1cm}
But as $\frac{\bar{n}_i^{k_0}}{N^k}\to 0$ as $k\to\infty$ we deduce that $\limsup_k\frac{\bar{n}_i^k}{N^k}\leq
q_i^*+\varepsilon$. Since this is true for any $\varepsilon$, and
$\frac{k}{N^k}\rightarrow 0$, we can conclude that $\limsup_k\frac{N_i^k}{N^k}\leq q_i^*$.
Now using the indexation on all the strata and the result for the strata
with variance zero, we deduce that for $1\leq
i\leq I$,
$$
\begin{array}{lll}
\liminf_k\dfrac{N_i^k}{N^k}
=\liminf_k\Big(1-\sum_{\stackrel{j=1}{j\neq i}}^I\dfrac{N_j^k}{N^k}\Big)&\geq& 1-\sum_{\stackrel{j=1}{j\neq i}}^I\limsup_k\dfrac{N_j^k}{N^k}\\&=&1-\sum_{\stackrel{j=1}{j\neq i}}^Iq_j^*=q_i^*.
\end{array}
$$
This concludes the proof.\end{proof}

\vspace{0.3cm}

\vspace{0.5cm}

\section{Numerical examples and applications to option pricing}
\label{snum}

\subsection{A first simple example}
\label{ssfirstex}

We compute $c=~\m{esp}X$ where $X\sim\m{nor}(0,1)$.

\vspace{0.2cm}

Let $I=10$. We choose the strata to be given by the $\alpha$-quantiles $y_\alpha$ of the normal law
for $\alpha=i/I$ for $1\leq i\leq I$. That is to say  $A_i=(y_{\frac{i-1}{I}},y_{\frac{i}{I}}]$ for all
$1\leq i\leq I$, with the convention that $y_0=-\infty$ and $y_1=+\infty$.

\vspace{0.2cm}

In this setting we have $p_i=1/10$ for all $1\leq i\leq I$.

\vspace{0.3cm}

Let us denote by $d(x)$ the density of the law $\m{nor}(0,1)$. Thanks to the relation
$d'(x)=-xd(x)$ and using integration by parts, we can establish that, for all $1\leq i\leq I$,
$$\m{esp}\Big(X\m{ind}_{y_{\frac{i-1}{I}}<X\leq{y_{\frac{i}{I}}}}\Big)=d(y_{\frac{i-1}{I}})-d(y_{\frac{i}{I}}),$$
and
$$\m{esp}\Big(X^2\m{ind}_{y_{\frac{i-1}{I}}<X\leq{y_{\frac{i}{I}}}}\Big)
=y_{\frac{i-1}{I}}d(y_{\frac{i-1}{I}})-y_{\frac{i}{I}}d(y_{\frac{i}{I}})+p_i,$$
with the convention that $y_0d(y_0)=y_1d(y_1)=0$.

We can then compute the exact $\sigma_i^2=\m{var}(X|X\in A_i)$'s and the optimal standard deviation of the non-adaptive stratified estimator,
$$\sigma_*=\sum_{i=1}^Ip_i\sigma_i\simeq 0.1559335$$

We can also for example compute
$$q_5^*=0.04685$$

This will give us benchmarks for our numerical tests.

\vspace{0.4cm}

We will compute $\hat{c}^k$ for $k=1,\ldots,4$. We choose $N^1=300$, $N^2=1300$,
$N^3=11300$ and $N^4=31300$.

First for one realization of the sequence $(\hat{c}^k)_{k=1}^4$ we plot the evolution of $\frac{N_5^k}{N^k}$, when we use procedure {\it a)} or {\it b)} for the computation of allocations.  This is done on Figure \ref{fignink}.

 \begin{figure}
  \begin{center}
 \includegraphics[width=9cm]{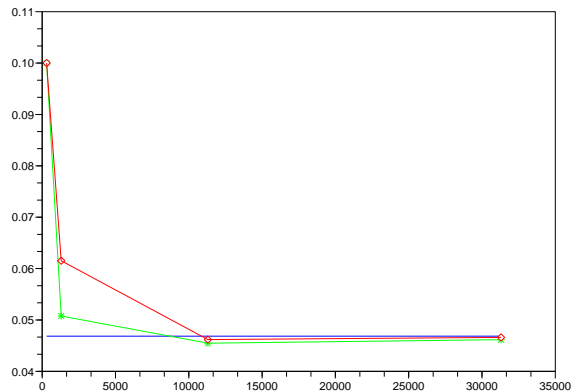}
 \caption{Successive values of $\frac{N_5^k}{N^k}$ for $1\leq k\leq 4$, for procedure {\it a)}
  (the $\diamond$-line) and procedure {\it b)} (the $*$-line), in function of $N^k$. The horizontal line is at level
 $q_5^*$.}
 \label{fignink}
 \end{center}
 \end{figure}

We observe that the convergence of $\frac{N_5^k}{N^k}$ to $q_5^*$ is faster
with procedure {\it b)}.

\vspace{0.4cm}

Second, to estimate the variance of our adaptive stratified estimator, we do $L=10000$ runs of all the procedure leading to the sequence $(\hat{c}^k)_{k=1}^4$. For $1\leq k\leq 4$ we compute,
$$\hat{v}^k=\frac{1}{L}\sum_{l=1}^L([\hat{c}^k]^l)^2-\Big(\frac{1}{L}\sum_{l=1}^L[\hat{c}^k]^l\Big)^2,$$
with the $\big([\hat{c}^k]^l\big)_{1\leq l\leq L}$ independent runs of the algorithm till step $k$.
This estimates the variance of the stratified estimator at step $k$ ($N^k$ total drawings have been used). To compare with $\sigma_*$ we compute the quantities
$$\hat{s}_k=\sqrt{N^k\hat{v}^k}$$ (in other words we compare the standard deviation of our adaptive stratified estimator with $N^k$ total drawings with the one of the non-adaptive stratified estimator with optimal allocation, for the same number of total drawings).

 \begin{figure}
 \begin{center}
  \includegraphics[width=9cm]{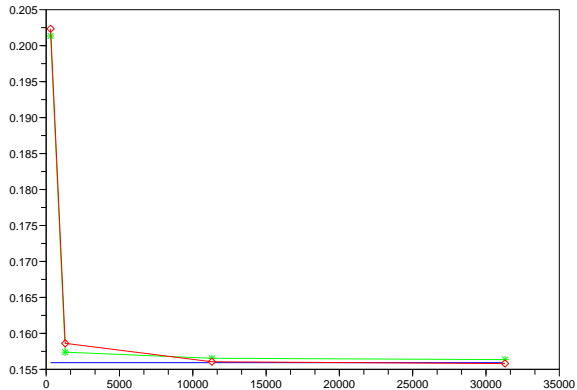}
 \caption{Successive values of $\hat{s}^k$ for $1\leq k\leq p$, for procedure {\it a)}
  (the $\diamond$-line) and procedure {\it b)} (the $*$-line), in function of $N^k$ (the abscissas axe). The horizontal line is at level
 $\sigma_*$.}
 \label{figsk} 
 \end{center}
 \end{figure}

The values are ploted on Figure \ref{figsk}.
We observe that the convergence to $\sigma_*$ is slightly faster with procedure {\it b)}. This corresponds to the fact that the convergence of the $\frac{N_i^k}{N^k}$'s is faster with this later procedure (see Proposition \ref{prop-vit-conv}).

\vspace{0.5cm}
We wish to compare the efficiency of our algorithm with the one of the non-adaptive stratified estimator with proportional allocation. Indeed this is the one we would use if we did not know the $\sigma_i$'s.

With the same strata as in the previous setting the stratified estimator with proportional allocation of $c$ for a total number of drawings $N^4=31300$ is
$$\bar{c}=\frac{1}{N^4}\sum_{i=1}^{10}\sum_{j=1}^{3130}X_i^j.$$
We will compare it to $\hat{c}^4$ that was computed in the example above. As we have seen in the Introduction, the variance of $\bar{c}$ is
$$\frac{1}{N^4}\sum_{i=1}^{10}p_i\sigma_i^2.$$

We do $L=10000$ runs of $\hat{c}^4$ and $\bar{c}$. We get an estimation $\hat{v}^4$ of the variance
of $\hat{c}^4$ as previously. In a similar manner we get an approximation 
$\bar{v}=\frac{1}{L}\sum_{l=1}^L([\bar{c}]^l)^2-\Big(\frac{1}{L}\sum_{l=1}^L[\bar{c}]^l\Big)^2$ of the variance of $\bar{c}$. 

As $\sum_{i=1}^{10}p_i\sigma_i^2\geq \big(\sum_{i=1}^{10}p_i\sigma_i \big)^2$ we know that we will have $\bar{v}\geq\hat{v}^4$. But to compute $\hat{c}^4$ we do some additional computations compared to a non adaptive stratified estimator. This has a numerical cost.
We thus use the $L$ runs to compute the average computation times $\hat{t}^4$ and $\bar{t}$, respectively of $\hat{c}^4$ and $\bar{c}$.

We have $\hat{t}^4\hat{v}^4=6.29*10^{-8}$ and $\bar{t}\bar{v}=7.57*10^{-8}$. This means that in this toy example
the numerical cost of our algorithm is not that much balanced by the achieved variance reduction.

\subsection{Applications to option pricing}\label{ssopt}

\subsubsection{The setting}

We wish to compare our results with the ones of \cite{glass}.

We will work on the example of the arithmetic Asian option in the Black-Scholes model  presented in this paper. We shortly present the setting. We have a single underlying asset, with price at time $t$ denoted by $S_t$. Under the risk neutral measure $\m{P}$, the price $(S_t)_t$ follows the stochastic differential equation,
$$dS_t=V\,S_tdW_t+rS_tdt,$$
with $r$ the constant interest rate, $V$ the constant asset's volatility, $W_t$ a standard Wiener
process, and $S_0$ fixed. 

Let $T>0$ be the option's maturity and  $\big(t_m=\frac{mT}{d}\big)_{1\leq m\leq d}$  the sequence of times when the value of the underlying asset is monitored to compute the average.
The discounted payoff of the arithmetic Asian option with strike $K$ 
is given by
$$e^{-rT}\Big(\frac{1}{d}\sum_{m=1}^dS_{t_m}-K \Big)^+.$$
Thus the price of the option is given by
$$c=\m{esp}\Big[ \,e^{-rT}\Big(\frac{1}{d}\sum_{m=1}^dS_{t_m}-K \Big)^+\,\Big].$$

But in this Black-Scholes setting we can exactly simulate the $S_{t_m}$'s using the fact that
$S_{t_0}=S_0$ and
\begin{equation}
\label{brogeo}
S_{t_m}=S_{t_{m-1}}\exp\big([r-\frac{1}{2}V^2](t_m-t_{m-1})+V\sqrt{t_m-t_{m-1}}X^m \big),
\quad\forall 1\leq m\leq d,
\end{equation}

where $X^1,\ldots,X^d$ are independent standard normals. Thus,
$$c=\m{esp}[g(X)\m{ind}_D(X)],$$
with $g$ some deterministic function, $D=\{x\in\m{R}^d:\,g(x)>0\}$, and $X$ a
$\m{R}^d$-valued random variable with law
$\m{nor}(0,I_d)$.

\vspace{0.3cm}

In \cite{glass} the authors discuss and link together two issues: importance sampling and stratified sampling.

Their importance sampling technique consists in a change of mean of the gaussian vector $X$.
Let us denote by $h(x)$ the density of the law $\m{nor}(0,I_d)$ and by $h_\mu(x)$ the density
of the law $\m{nor}(\mu,I_d)$ for any $\mu\in\m{R}^d$. We have,
$$c=\int_Dg(x)\frac{h(x)}{h_\mu(x)}h_\mu(x)dx=\m{esp}[g(X+\mu)\frac{h(X+\mu)}{h_\mu(X+\mu)}\m{ind}_D(X+\mu)].$$
 The variance of $g(X+\mu)\frac{h(X+\mu)}{h_\mu(X+\mu)}\m{ind}_D(X+\mu)$ is given by
$$\int_D\Big(g(x)\frac{h(x)}{h_\mu(x)}-c\Big)^2h_\mu(x)dx.$$
Heuristically, this indicates that an effective choice of $h_\mu$ should give weight to points for which the product of the payoff and the density is large. In other words, if we define $G(x)=\log g(x)$ we should look for $\mu\in\m{R}$ that verifies,
\begin{equation}
\label{muargmax}
 \mu=\argmax_{x\in D}\,\Big(\, G(x)-\frac{1}{2}x'x\,\Big)
\end{equation}

The most significant part of the paper \cite{glass} is aimed at giving an asymptotical sense to
this heuristic, using large deviations tools.

The idea is then to sample $g(X+\mu)\frac{h(X+\mu)}{h_\mu(X+\mu)}\m{ind}_D(X+\mu)$.

\vspace{0.2cm}

Standard computations show that for any $\mu\in\m{R}^d$,
$$
c=\m{esp}\big[g(X+\mu)e^{-\mu'X-(1/2)\mu'\mu}\m{ind}_D(X+\mu) \big].
$$
Thus the problem is now to build a Monte Carlo estimator of $c=\m{esp}f_\mu(X)$, sampling
$f_\mu(X)$ with $X\sim\m{nor}(0,I_d)$, and with $f_\mu(x)=g(x+\mu)e^{-\mu'x-(1/2)\mu'\mu}\m{ind}_D(x+\mu)$,
for the vector $\mu$ satisfying \eqref{muargmax}.

\vspace{0.3cm}
The authors of \cite{glass} then propose to use a stratified estimator of $c=\m{esp}f_\mu(X)$.
Indeed for $u\in\m{R}^d$ with $u'u=1$, and $a<b$ real numbers, it is easy to sample according to the conditional law of $X$ given 
$u'X\in[a,b]$. 

It can be done in the following way (see Subsection 4.1 of \cite{glass} for details).
We first sample $Z=\Phi^{-1}(V)$ with $\Phi^{-1}$ the inverse of the cumulative normal distibution, and $V=\Phi(a)+U(\Phi(b)-\Phi(a))$, with $U$ uniform on $[0,1]$. Second we sample $Y\sim\m{nor}(0,I_d)$ independent of $Z$.
We then compute,
$$X=uZ+Y-u(u'Y),$$
which by contruction has the desired conditional law.

\vspace{0.2cm} 

Let be $u\in\m{R}^d$ satisfy $u'u=1$. With our notation the stratified estimator $\widehat{c}$ in \cite{glass} is built in the following way. They take $I=100$.
As in subsection \ref{ssfirstex} we denote by $y_\alpha$ the $\alpha$-quantile of the law $\m{nor}(0,1)$. For all $1\leq i\leq I$, they take 
$A_i=\{x\in\m{R}^d: y_{\frac{i-1}{I}}<u'x\leq y_{\frac{i}{I}}\}$.
That is to say $X_i$ has the conditional law of $X$ given $y_{\frac{i-1}{I}}<u'X\leq y_{\frac{i}{I}}$, for all $1\leq i\leq I$.
As in this setting $u'X\sim\m{nor}(0,1)$, they have $p_i=1/I$ for all $1\leq i\leq I$.

They then do proportional allocation, that is to say, $N_i=p_iN$ for all $1\leq i\leq I$, where $N$ is the total number of drawings (in other words $q_i=p_i$). Then, the variance of their stratified estimator is
$$\frac{1}{N}\sum_{i=1}^Ip_i\sigma_i^2.$$
According to the Introduction, that choice ensures variance reduction.

\vspace{0.2cm}

The question of the choice 
of the projection direction $u$ arises. The authors take
$u=\mu/(\mu'\mu)$, with the vector $\mu$ satisfying \eqref{muargmax} that has been used
for the importance sampling. They claim that this provides in practice a very efficient projection direction, for their stratified estimator with proportional allocation.

As $\big(\sum_{i=1}^Ip_i\sigma_i\big)^2\leq\sum_{i=1}^Ip_i\sigma_i^2$ (i.e. proportional allocation is suboptimal),
 if $u$ is a good projection direction for a stratified estimator with proportional allocation, it is a good direction  for a stratified estimator with optimal allocation.

\vspace{0.3cm}
In the sequel we take the same direction $u$ and the same strata as in \cite{glass}, and discuss allocation. Indeed we may wish to do optimal allocation and take
$q_i=q_i^*=\frac{p_i\sigma_i}{\sum_jp_j\sigma_j}$. The trouble is the analytical computation of the quantities
$$\sigma_i^2=\m{var}(f_\mu(X)|u'X\in(y_{\frac{i-1}{I}},y_{\frac{i}{I}}]),$$
is not tractable, at least when $f_\mu$ is not linear. As the $p_i$'s are known, this is exactly the kind of situation where our adaptive stratified estimator can be useful. 

\subsubsection{The results}

In all the tests we have taken $S_0=50$, $V=0.1$, $r=0.05$ and $T=1.0$. The total number of drawings is $N=1000000$.

We call GHS the procedure used in \cite{glass}, that is importance sampling plus stratified sampling with proportional allocation. We call SSAA our procedure, that is the same importance sampling plus stratified sampling with adaptive allocation.

More precisely in the procedure SSAA we choose $N^1=100000$, $N^2=400000$, $N^3=500000$ and compute
our adaptive stratified estimator $\hat{c}^3$ of $c=\m{esp}f(X)$, with the same strata as in GHS. We have used procedure {\it a)} for the computation of allocations. We denote by $\bar{c}$ the GHS estimator of $c$.

\vspace{0.3cm}

We call <<variance GHS>> or <<variance SSAA>> the quantity
$\widehat{\sigma}$, which is an estimation of the variance of $\bar{c}$ or $\hat{c}^3$.
More precisely for GHS,
$$(\widehat{\sigma})^2=\frac{1}{N}\sum_{i=1}^Ip_i\widehat{\sigma_i}^2,$$
where for each $1\leq i\leq I$,
$$\widehat{\sigma_i}^2=\frac{1}{p_iN}\sum_{j=1}^{p_iN}f^2(X_i^j)
-\Big(\frac{1}{p_iN}\sum_{j=1}^{p_iN}f(X_i^j)\Big)^2,$$
and for SSAA
$$(\widehat{\sigma})^2=\frac{1}{N}\Big(\sum_{i=1}^Ip_i\widehat{\sigma_i}\Big)^2,$$
where for each $1\leq i\leq I$,
$$(\widehat{\sigma_i})^2=\frac{1}{N_i^3}\sum_{j=1}^{N_i^3}f^2(X_i^j)
-\Big(\frac{1}{N_i^3}\sum_{j=1}^{N_i^3}f(X_i^j)\Big)^2.$$

\begin{table}
 \begin{center}
\begin{tabular}{ccccc}
\hline
d &K &Price&variance SSAA&ratio GHS/SSAA\\
\\
\hline
16& 45& 6.05& $2.37\times 10^{-8}$& 2.04\\
& 50& 1.91& $1.00\times 10^{-7}$& 35\\
& 55& 0.20&  $5.33\times 10^{-9}$ &  39.36 \\
\\
64& 45& 6.00& $3.36\times 10^{-9}$  &   3.34\\
& 50& 1.84& $9.00\times 10^{-10}$   &  1.60 \\
& 55& 0.17& $6.40\times 10^{-9}$   &  61  \\
\hline
\end{tabular}
\caption{Results for a call option with $S_0=50$, $V=0.1$, $r=0.05$, $T=1.0$ and $N=1000000$ (and $I=100$).}
\label{tab1} 
\end{center}
\end{table}

\vspace{0.5cm}

Tables \ref{tab1} and \ref{tab2} show the results respectively for a call option and a put option. We call <<ratio GHS/SSAA>> the variance GHS divided by the variance SSAA. In general the improvement is much better for a put option. Indeed the variance is often divided by $100$ in this case.

\vspace{1cm}
 \begin{table}
 \begin{center}
\begin{tabular}{ccccc}
\hline
d &K &Price&variance SSAA&ratio GHS/SSAA\\
\\
\hline
16& 45& 0.013& $7.29\times 10^{-10}$& 107\\
& 50& 0.63& $7.29\times 10^{-8}$& 79\\
& 55& 3.74& $2.50\times 10^{-5}$& 249\\
\\
64& 45& 0.011& $5.76\times 10^{-10}$& 95\\
& 50& 0.62& $5.61\times 10^{-8}$& 64\\
& 55& 3.69& $1.85\times 10^{-5}$& 58\\
\hline
\end{tabular}
\caption{Results for a put option with $S_0=50$, $V=0.1$, $r=0.05$, $T=1.0$ and $N=1000000$ (and $I=100$).}
\label{tab2} 
\end{center}
\end{table}

A further analysis can explain these results. We plot on Figure \ref{figcall} and \ref{figput} the values of the $\widehat{\sigma}_i$'s and the estimated values of the conditional expectations $\m{esp}f_\mu(X_i)$'s, for a call and a put option, with $d=64$ and $K=45$, a case for which the ratio GHS/SSAA is 3.34 in the call case and
95 in the put case.

\begin{figure}
 \begin{center}
  \includegraphics[width=6cm]{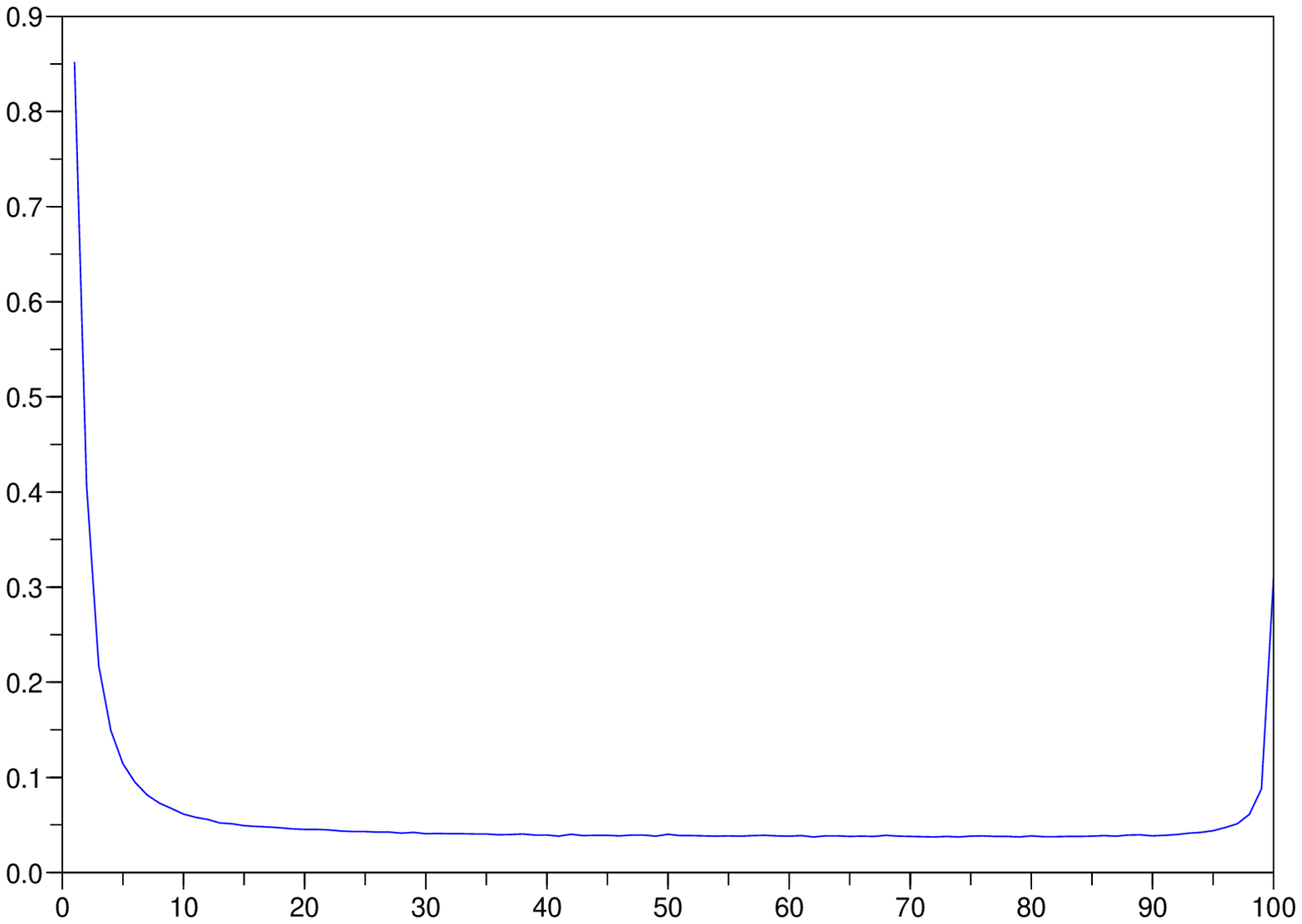}
\includegraphics[width=6cm]{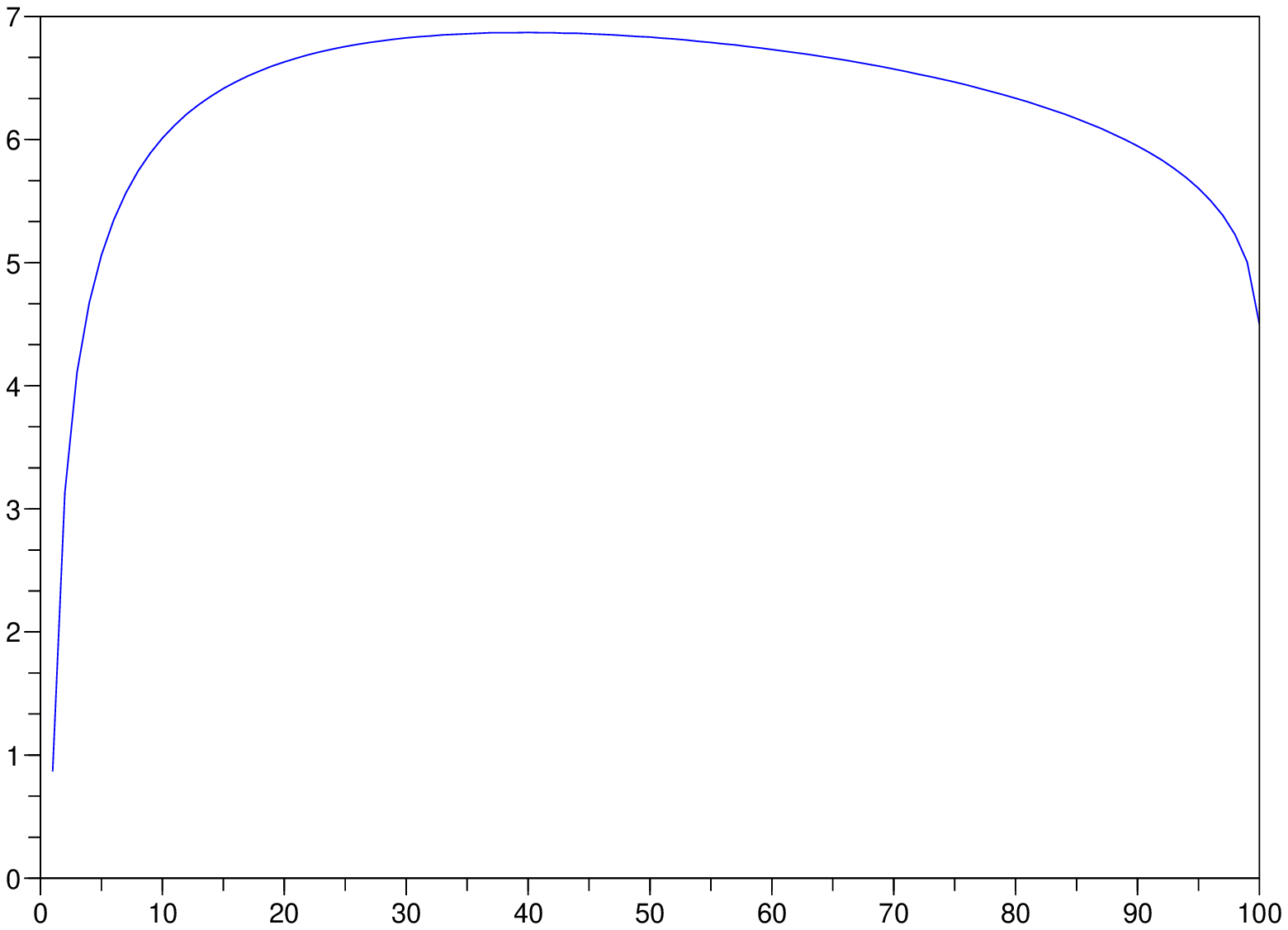}
 \caption{On the left: value of $\widehat{\sigma}_i$ in function of the stratum index $i$ in the case of a call option. On the right: estimated value of $\m{esp}f_\mu(X_i)$. (Parameters are the same as in Tables \ref{tab1}, with $d=64$ and $K=45$).}
 \label{figcall} 
 \end{center}
 \end{figure}

\begin{figure}
 \begin{center}
  \includegraphics[width=6cm]{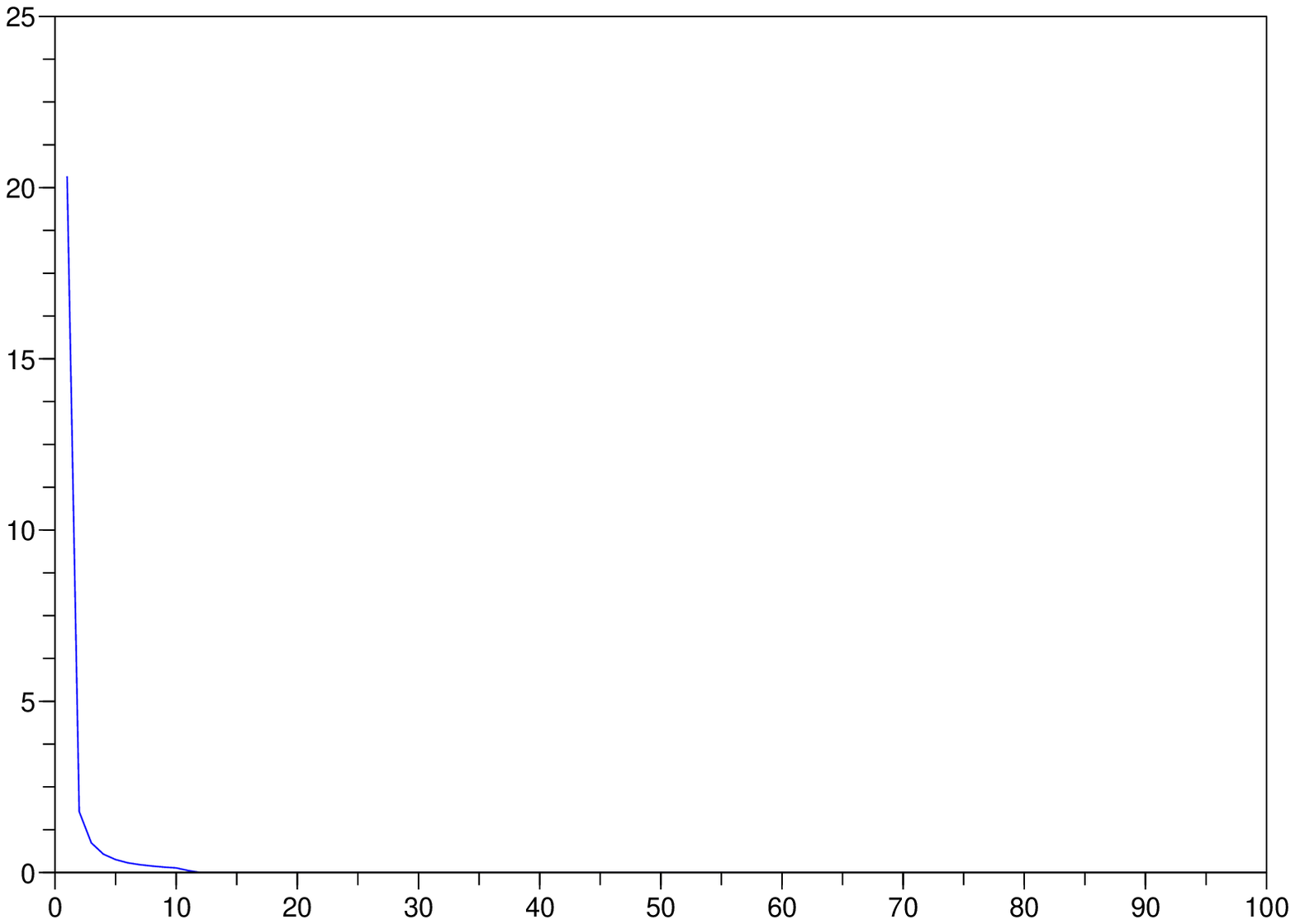}
\includegraphics[width=6cm]{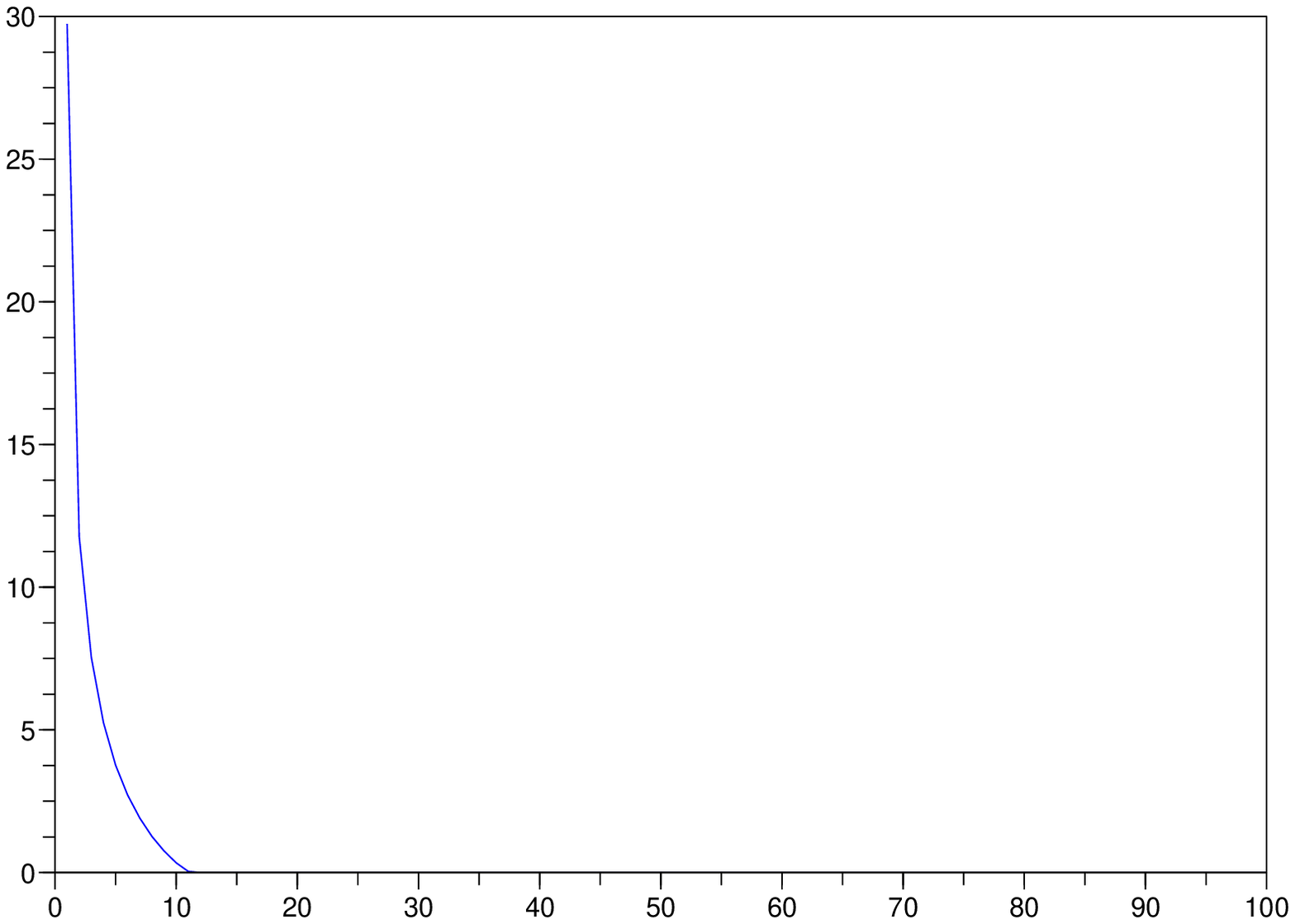}
 \caption{On the left: value of $\widehat{\sigma}_i$ in function of the stratum index $i$ in the case of a put option. On the right: estimated value of $\m{esp}f_\mu(X_i)$. (Same parameters than in Figure \ref{figcall}).}
 \label{figput} 
 \end{center}
 \end{figure}

We observe that in the case of the put option the estimated conditional variance of about $90\%$ of the strata is zero, unlike in the case of the call option. These estimated conditional variances are zero, because in the corresponding strata the estimated conditional expectations are constant with value zero. 

But these strata are of non zero probability (remember that in this setting $p_i=0.01$, for all $1\leq i\leq 100$). Thus the GHS procedure with proportional allocation will invest drawings in these strata, resulting in a loss of accuracy, while in our SSAA procedure most of the drawings are made in the strata of non zero estimated variance. 

\vspace{0.1cm}

One can wonder if the expectation in the strata of zero observed expectation is really zero, or if it is just a numerical effect. 
We define the deterministic function $s:\m{R}^d\to\m{R}$ by
$$
s(x)=\frac{S_0}{d}\sum_{m=1}^d\exp\Big(\sum_{p=1}^m\big\{ [r-\frac{V^2}{2}]\frac{T}{d}+V\sqrt{\frac{T}{d}}x^p\big\} \Big),\quad\forall x=(x^1,\ldots,x^d)'\in\m{R}^d.
$$
With the previous notations, in the put option case, we have $f_\mu(X_i)=0$ a.s.,
and thus $\m{esp}f_\mu(X_i)=0$, if $s(X_i+\mu)\geq K$ a.s. (note that $i$ denotes here the stratum index and not the component of the random vector $X_i$).

Thus the problem is to study, in function of $z\in\m{R}$, the deterministic values of  $s(x+\mu)$ for $x\in\m{R}^d$ satisfying $u'x=z$. The following facts can be shown. Whatever the value of $u$ or $z$ the quantity $s(x+\mu)$ has no upper bound. Thus in the call option case
no conditional expectation $\m{esp}f_\mu(X_i)$ will be zero. To study the problem of the lower bound we denote by $M$ the matrix of size $d\times d$ given~by
$$M=\begin{pmatrix}
1& 0 &\ldots & 0\\
1& 1& \ddots& \vdots\\
\vdots&&\ddots&0\\
1&\ldots&\ldots&1\\
         \end{pmatrix},\quad
\text{with inverse}\quad
M^{-1}=\begin{pmatrix}
1& 0 &\ldots & 0\\
-1& 1& \ddots& \vdots\\
\vdots&\ddots&\ddots&0\\
0&\ldots&-1&1\\
         \end{pmatrix},
$$
and by $\m{ind}$ the $d$-sized vector $(1,\ldots,1)'$. If we use the change of variable
$$y=M\Big( [r-\frac{V^2}{2}]\frac{T}{d}\m{ind}+V\sqrt{\frac{T}{d}}(x+\mu) \Big),
$$
we can see that
 minimizing 
 $s(x+\mu)$ for $x\in\m{R}^d$ satisfying $u'x=z$ is equivalent to minimizing
$\frac{S_0}{d}\sum_{m=1}^d\exp(y^m)$ for $y\in\m{R}^d$ satisfying
\begin{equation}
\label{contrainte}
w'y=v,
\end{equation}
where,
$$w=(M^{-1})'u,$$
and
$$v=u'\Big( [r-\frac{V^2}{2}]\frac{T}{d}\m{ind}+V\sqrt{\frac{T}{d}}(x+\mu) \Big)=V\sqrt{\frac{T}{d}}(z+u'\mu)+(r-\frac{V^2}{2})\sum_{m=1}^du_m.$$

If all the components of $w$ are stricly positive the lower bound of $s(x+\mu)$
under the constraint $u'x=z$ is
\begin{equation}
\label{fet}
s^*=\frac{S_0}{d}\times\exp\Big( \frac{v-\sum_{m=1}^dw_m\log w_m}{\sum_{m=1}^dw_m}\Big)\times
\sum_{m=1}^dw_m.
\end{equation}
If all the components of $w$ are stricly negative we get the same kind of result by a change of sign. Otherwise the lower bound is zero: it is possible to let the $y^m$'s tend to $-\infty$
with \eqref{contrainte} satisfied.

In the numerical example that we are analysing the direction vector $u$ is the same in the call or put option cases, and its components are stricly positive and decreasing with the index (see Figure \ref{figum}). Thus the components of $w$ are strictly positive and the lower bound is given by $s^*$ defined by \eqref{fet}. With $z$ taking values in the $90$ last strata we have $s^*>45$. Thus the conditional expectations $\m{esp}f_\mu(X_i)$ are truly
zero in these strata.

\begin{figure}
 \begin{center}
  \includegraphics[width=8cm]{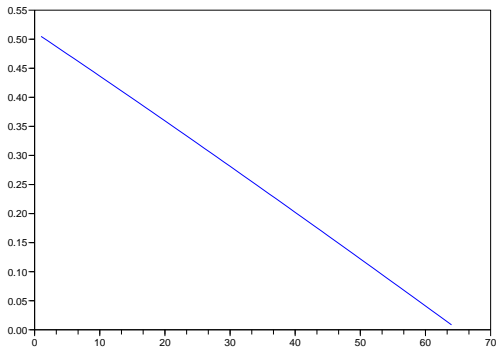}
 \caption{Value of the component $u_m$ of $u\in\m{R}^d$ in function of $m$.}
 \label{figum} 
 \end{center}
 \end{figure}

We can then wonder if it is worth stratifying the part of the real line corresponding to these strata, in other words stratifying $\m{R}^d$ and not only $D$. Maybe stratifying $D$
and making proportional allocation will provide a sufficient variance reduction. But this would require a first analysis, while our SSAA procedure avoids automatically to make
a large number of drawings in $D^c$.

\vspace{0.3cm}

To conclude on the efficiency of our algorithm in this example let us notice that the computation times of the GHS and SSAA procedures are nearly the same (less than $1\%$ additional time for the SSAA procedure). Indeed, unlike in the toy example of Subsection \ref{ssfirstex}, the computation time of the allocation of the drawings in the strata is almost negligible in comparison to the other calculations (drawings etc...).


\section{Appendix}

We justify the use of procedure {\it b)}  in the following proposition.

\begin{prop}
\label{prop-opt}
 






When $\widehat{\sigma}_i^{k-1}>0$ for some $1\leq i\leq I$, by computing at Step $k$ the $m_i^k$'s with the procedure {\it b)}
described in Section \ref{sec_algo}, we find $(m_1^k,\ldots,m_I^k)\in\m{R}_+^I$ that minimizes
$$\sum_{i=1}^I\dfrac{p_i^2(\widehat{\sigma}_i^{k-1})^2}{N_i^{k-1}+1+m_i^k},$$
under the constraint $\sum_{i=1}^Im_i^k=N^k-N^{k-1}-I$.

\end{prop}

\begin{proof} 
First note that if $\widehat{\sigma}_i^{k-1}=0$ for some index $i$ it is clear that we have to set $m_i^k=0$ and to rewrite the minimization problem for the indexes corresponding to
$\widehat{\sigma}_i^{k-1}>0$. This corresponds to the very beginning of procedure {\it b)}.

For the seek of simplicity, and without loss of generality, we consider in the sequel
that $\widehat{\sigma}_i^{k-1}>0$ for all $1\leq i\leq I$, and thus work with the indexation
$\{1,\ldots,I\}$.

We will note $M=N^k-N^{k-1}-I$, and, for all $1\leq i\leq I$, $n_i=N_i^{k-1}+1$, $\alpha_i=p_i\widehat{\sigma}_i^{k-1}$, and $m_i=m_i^k$. We thus seek $(m_1,\ldots,m_I)\in\m{R}_+^I$
that minimizes $\sum_{i=1}^I\frac{\alpha_i^2}{n_i+m_i}$ under the constraint
$\sum_{i=1}^Im_i=M$.

\vspace{0.2cm}

{\it Step 1: Lagrangian computations.}
We write the Lagrangian corresponding to our minimization problem, for all $(m,\lambda)\in\m{R}_+^I\times\m{R}$:

\begin{equation*}
 \m{L}(m,\lambda)=\sum_{i=1}^I\frac{\alpha_i^2}{n_i+m_i}+\lambda(\sum_{i=1}^Im_i-M)
=\sum_{i=1}^Ih_i(m_i,\lambda)-\lambda M.
\end{equation*}
with $h_i(x,\lambda)=\Big(\frac{\alpha_i^2}{n_i+x}+\lambda x\Big)$ for all $i$.

We first minimize $\m{L}(m,\lambda)$ with respect to $m$ for a fixed $\lambda$.

For any $\lambda\in\m{R}$ let us denote $m(\lambda):=\argmin_{m\in\m{R}_+^I}\m{L}(m,\lambda)$.

Minimizing $\m{L}(m,\lambda)$ with respect to $m$ is equivalent to minimizing
$h_i(m_i,\lambda)$ with respect to $m_i$ for all $i$.

 The derivative of each $h_i(.,\lambda)$ has the same sign as 
$-\alpha_i^2+\lambda(n_i+x)^2$. 

If $\lambda\leq 0$ we have $m(\lambda)=(\infty,\ldots\infty)$. 

If $\lambda>0$ there are two cases to consider for each $h_i$: 

\begin{equation}
 \label{mi}
\begin{array}{c}
\text{either }\;\; \lambda>\frac{\alpha_i^2}{n_i^2}  \text{ and }m_i(\lambda)=0,\\
\\
\text{or } 
\lambda\leq\frac{\alpha_i^2}{n_i^2} \text{ and }m_i(\lambda)=\sqrt{\alpha_i^2/\lambda}-n_i.\\
\end{array}
\end{equation}

 To sum up we have
\begin{equation*}
 \m{L}(m(\lambda),\lambda)=\left\{
\begin{array}{ll}
 -\infty &\text{ if }\lambda<0,\\
\\
0&\text{ if }\lambda=0,\\
\\
\sum_{i=1}^I\Big[\m{ind}_{\{\lambda>\frac{\alpha_i^2}{n_i^2}\}}\frac{\displaystyle\alpha_i^2}
{\displaystyle n_i}
+ \m{ind}_{\{\lambda\leq\frac{\alpha_i^2}{n_i^2}\}}(2\alpha_i\sqrt{\lambda}-n_i\lambda)\Big]
-M\lambda&\text{ if }\lambda>0.
\end{array}
\right.
\end{equation*}

We now look for $\lambda^*$ that maximizes $\m{L}(m(\lambda),\lambda)$. For all
$\lambda\in (0,\infty)$ we have,

\begin{equation}
\label{derivative} 
\partial_\lambda\m{L}(m(\lambda),\lambda)=\sum_{i=1}^I\m{ind}_{\{\lambda\leq\frac{\alpha_i^2}{n_i^2}\}}\big(\frac{\alpha_i}{\sqrt{\lambda}}-n_i\big)-M.
\end{equation}
This function is continuous on $(0,+\infty)$, equal to $-M$ for $\lambda\geq
\max_i\frac{\alpha_i^2}{n_i^2}$, decreasing on
$(0,\max_i\frac{\alpha_i^2}{n_i^2}]$ and tends to $+\infty$ as $\lambda$ tends to
$0$. We deduce that
 $\lambda\mapsto\m{L}(m(\lambda),\lambda)$ reaches its unique maximum at some $\lambda^*\in(0,\max_i\frac{\alpha_i^2}{n_i^2})$. 

If $\partial_\lambda\m{L}\Big(m\big(\frac{\alpha_{(i)}^2}{n_{(i)}^2}\big),\frac{\alpha_{(i)}^2}{n_{(i)}^2}\Big)<0$ for all $1\leq i\leq I$, we set $i^*=0$. 

Otherwise we sort in increasing order the $\alpha_i^2/n_i^2$'s, index with  $(i)$ the ordered quantities, and note $i^*$ the integer such that
\begin{equation}
\label{grad-lag}
\partial_\lambda\m{L}\Big(m\big(\frac{\alpha_{(i^*)}^2}{n_{(i^*)}^2}\big),\frac{\alpha_{(i^*)}^2}{n_{(i^*)}^2}\Big)\geq 0\quad\text{and}\quad
\partial_\lambda\m{L}\Big(m\big(\frac{\alpha_{(i^*+1)}^2}{n_{(i^*+1)}^2}\big),\frac{\alpha_{(i^*+1)}^2}{n_{(i^*+1)}^2}\Big)<0.
\end{equation}

Then $\lambda^*$ belongs to $\big[\frac{\alpha_{(i^*)}^2}{n_{(i^*)}^2},\frac{\alpha_{(i^*+1)}^2}{n_{(i^*+1)}^2}\big)$,
or $\big(0,\frac{\alpha_{(1)}^2}{n_{(1)}^2}\big)$ if $i^*=0$.
But on this interval 
$$\partial_\lambda\m{L}(m(\lambda),\lambda)=\sum_{j=i^*+1}^I(\frac{\alpha_{(j)}}{\sqrt{\lambda}}-n_{(j)})-M.$$
As $\partial_\lambda\m{L}(m(\lambda^*),\lambda^*)=0$ we have,
$$\frac{1}{\sqrt{\lambda^*}}=\frac{\displaystyle M+\sum_{j=i^*+1}^In_{(j)}}
{\displaystyle \sum_{j=i^*+1}^I\alpha_{(j)}}.$$

Clearly, if 
$i^*\neq 0$, $\lambda^*\geq\frac{\alpha_{(i)}^2}{n_{(i)}^2}$ is equivalent to $i\leq i^*$. If $i^*=0$ then $\lambda^*<\frac{\alpha_{(i)}^2}{n_{(i)}^2}$ for all $1\leq i\leq I$.
Thus, according to \eqref{mi}, we have $m_{(i)}(\lambda^*)=0$ if $i\leq i^*$, and if $i>i^*$,
\begin{equation}
\label{expmi}
 m_{(i)}(\lambda^*)=\alpha_{(i)}.\frac{\displaystyle M+\sum_{j=i^*+1}^In_{(j)}}{\displaystyle\sum_{j=i^*+1}^I\alpha_{(j)}}-n_{(i)}.
\end{equation}

We have thus found $(m(\lambda^*),\lambda^*)$ that satisfies
$$\m{L}(m(\lambda^*),\lambda^*)=\max_{\lambda\in\m{R}}\min_{m\in\m{R}_+^I}\m{L}(m,\lambda),$$
which implies that $\m{L}(m(\lambda^*),\lambda^*)\leq\m{L}(m,\lambda^*)$
for all $m\in\m{R}_+^I$. Besides \eqref{expmi} implies
$\sum_{i=1}^Im_i(\lambda^*)=M$ and $\m{L}(m(\lambda^*),\lambda^*)=\m{L}(m(\lambda^*),\lambda)$ for all $\lambda\in\m{R}$.
Therefore $(m(\lambda^*),\lambda^*)$ is a saddle point of the Lagrangian and 
$m(\lambda^*)$ solves the constrained minimization problem. 

\vspace{0.2cm}
{\it Step 2.} We now look for a criterion to find the index $i^*$
satifying \eqref{grad-lag}. If $i^*\neq 0$, we have the following
equivalences using the concavity of
$\lambda\mapsto\m{L}(m(\lambda),\lambda)$ and \eqref{derivative}
\begin{equation*}
 i\leq i^*\quad\Leftrightarrow\quad\partial_\lambda\m{L}(m(\frac{\alpha_{(i)}^2}{n_{(i)}^2}),\frac{\alpha_{(i)}^2}{n_{(i)}^2})\geq 0\quad\Leftrightarrow\quad\frac{n_{(i)}}{\alpha_{(i)}}\geq
\frac{\displaystyle M+\sum_{j=i+1}^In_{(j)}}{\displaystyle\sum_{j=i+1}^I\alpha_{(j)}}.
\end{equation*}
In the same manner,
$$i^*=0\quad\Leftrightarrow\quad\frac{n_{(i)}}{\alpha_{(i)}}<
\frac{\displaystyle M+\sum_{j=i+1}^In_{(j)}}{\displaystyle\sum_{j=i+1}^I\alpha_{(j)}}
,\;\forall 1\leq i\leq I.$$

The proof of Proposition \ref{prop-opt} in then completed: in Points i) and ii) of procedure {\it b)} we find the index $i^*$ mentionned in Step 1, using the criterion of Step 2. In Point iii) we compute the solution of the optimization problem using the results of Step 1.

\end{proof}

\end{document}